\documentclass[twocolumn, superscriptaddress, floatfix, nofootinbib, 10pt]{revtex4-2}

\usepackage[utf8]{inputenc}
\usepackage[english]{babel}
\usepackage[sort&compress]{natbib}

\usepackage{graphicx}
\usepackage{amsmath, amssymb, amsthm}
\usepackage{tensor}
\usepackage{braket}
\usepackage{cases}
\usepackage{dsfont}
\usepackage[dvipsnames]{xcolor}
\usepackage{hyperref}
\hypersetup{
	colorlinks = true,
	allcolors = RoyalPurple
}
\usepackage[percent]{overpic}
\usepackage{svg}
\usepackage{tikz}
\usepackage[export]{adjustbox}
\usepackage{leftidx}
\usepackage{BOONDOX-cal} % more calligraphic letters, in particular H
\usepackage{orcidlink}

\DeclareMathOperator{\Tr}{Tr}

\newtheorem{definition}{Definition}[]

\newtheorem{lemma}[definition]{Lemma}
\newtheorem{theorem}[definition]{Proposition}

\newcommand{\comm}[2]{\left[ #1 , #2 \right]}

\newcommand{\norm}[1]{\left\lVert #1 \right\rVert}

\newcommand{\change}[1]{{#1}}

\newcommand{\horiz}{
	\begin{tikzpicture}
		\draw (0,0)grid[step=2](4,2);
		\fill[fill = cyan, opacity=0.35] (0,0) rectangle ++(2,2);
		\fill[fill = yellow, opacity=0.5] (2,0) rectangle ++(2,2);
		\node at (1,1.5) {$\comm{X}{Z}$};
		\node at (1,.5) {$(\mu_x, \mu_y)$};
		\node at (3,1.5) {$Z$};
		\node at (3,.5) {$(\mu_x, \mu_y +1)$};
		
		\draw[yshift=0.5cm] (-2,2)grid[step=2](2,4);
		\fill[fill = yellow, opacity=0.5] (-2,2.5) rectangle ++(2,2);
		\fill[fill = cyan, opacity=0.35] (0,2.5) rectangle ++(2,2);
		\node at (-1,4) {$Z$};
		\node at (-1,3) {$(\mu_x, \mu_y - 1)$};
		\node at (1,4) {$\comm{X}{Z}$};
		\node at (1,3) {$(\mu_x, \mu_y)$};
		
		\draw (6,0) grid[step=2] ++(2,4);
		\fill[fill = cyan, opacity=0.35] (6,2) rectangle ++(2,2);
		\fill[fill = yellow, opacity=0.5] (6,0) rectangle ++(2,2);
		\node at (7,3.5) {$\comm{X}{Z}$};
		\node at (7,2.5) {$(\mu_x, \mu_y)$};
		\node at (7,1.5) {$Z$};
		\node at (7,.5) {$(\mu_x + 1, \mu_y)$};
		
		\draw[xshift=0.5cm] (8,2) grid[step=2] (10,6);
		\fill[fill = yellow, opacity=0.5] (8.5,4) rectangle ++(2,2);
		\fill[fill = cyan, opacity=0.35] (8.5,2) rectangle ++(2,2);
		\node at (9.5,5.5) {$Z$};
		\node at (9.5,4.5) {$(\mu_x - 1, \mu_y)$};
		\node at (9.5,3.5) {$\comm{X}{Z}$};
		\node at (9.5,2.5) {$(\mu_x, \mu_y)$};
	\end{tikzpicture}
}
\begin{document}
\title{Problem Specific Classical Optimization of Hamiltonian Simulation}

\author{Refik Mansuroglu \orcidlink{0000-0001-7352-513X}}
\email[]{Refik.Mansuroglu@fau.de}
\affiliation{Department of Physics, Friedrich-Alexander Universität Erlangen-Nürnberg (FAU), Staudtstraße 7, 91058 Erlangen}

\author{Felix Fischer \orcidlink{0009-0001-3040-8184}}
\affiliation{Department of Physics, Friedrich-Alexander Universität Erlangen-Nürnberg (FAU), Staudtstraße 7, 91058 Erlangen}

\author{Michael J. Hartmann \orcidlink{0000-0002-8207-3806}}
\affiliation{Department of Physics, Friedrich-Alexander Universität Erlangen-Nürnberg (FAU), Staudtstraße 7, 91058 Erlangen}
% \affiliation{Max-Planck Institute for the Science of Light, Erlangen, Germany}

\date{\today}

\begin{abstract}
    Nonequilibrium time evolution of large quantum systems is a strong candidate for quantum advantage. Variational quantum algorithms have been put forward for this task, but their quantum optimization routines suffer from trainability and sampling problems. Here, we present a classical pre-processing routine for variational Hamiltonian simulation that circumvents the need of a quantum optimization by expanding rigorous error bounds in a perturbative regime for suitable time steps. The resulting cost function is efficiently computable on a classical computer. We show that there always exists potential for optimization with respect to a Trotter sequence of the same order and that the cost value has the same scaling as for Trotter in simulation time and system size. Unlike previous work on classical pre-processing, the method is applicable to any Hamiltonian system independent of locality and interaction lengths. Via numerical experiments for spin-lattice models, we find that our approach significantly improves digital quantum simulations capabilities with respect to Trotter sequences for the same resources. For short times, we find accuracy improvements of more than three orders of magnitude for our method as compared to Trotter sequences of the same gate number. Moreover, for a given gate number and accuracy target, we find that the pre-optimization we introduce enables simulation times that are consistently more than 10 times longer for a target accuracy of 0.1\%.
\end{abstract}

\maketitle

\section{Introduction}
\begin{figure*}[t]
	\centering
	\includegraphics[scale=0.7]{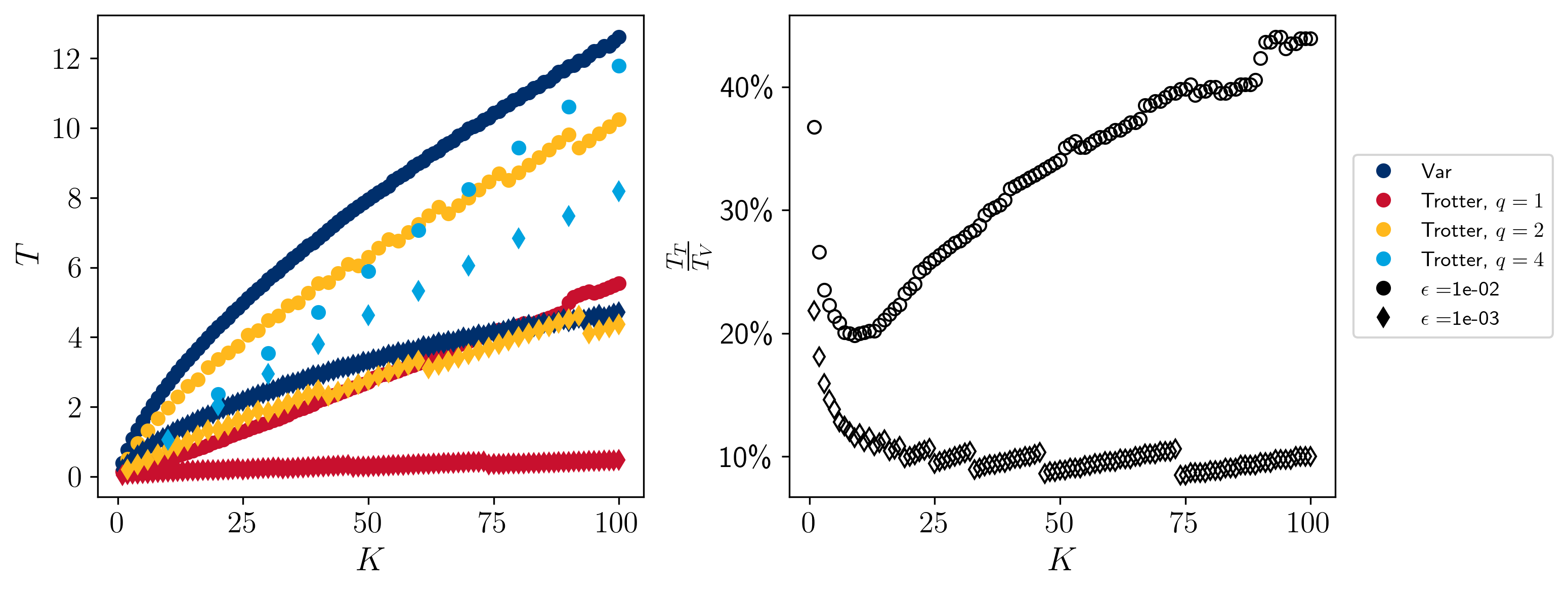}
	\caption{\textbf{Maximal simulation time $T$ of variational and Trotter sequences of order $q \in\{1,2,4\}$ (left) and simulation time ratio $T_T/T_V$ of first order sequences (right)} for a fixed gate count that is determined by $K$ and a threshold accuracy $\epsilon \in \{ 10^{-2}, 10^{-3} \}$ evaluated on a nearest neighbor XY model on a $3\times3$ quadratic lattice and random interaction strength $J^{(y)}_{\mu,\nu}$ centered around 0.5, $J^{(z)}_{\mu\nu}$ centered around 1 and $h_\mu=0.25$. The small jumps in $\frac{T_T}{T_V}$ are due to the finite step size in the search for optimal time steps $t$ to meet the error threshold $\epsilon$. See section \ref{sec:num} for details of the numerical experiments to produce the data shown here.}
	\label{fig:max_time}
\end{figure*}

Quantum simulation is believed to be one of the first applications in quantum computing that may show practical advantages over their classical counterparts \cite{Smith_2019, Arute2019}. To compile a digital quantum algorithm for simulating a time evolution $U(t)$ that will be typically generated by a $p$-local Hamiltonian,
\begin{align}
	U(t) = e^{-itH} \qquad \text{with} \qquad H = \sum_j c_j H_j,
\end{align}
a Trotter-Suzuki decomposition \cite{Suzuki_85}, that approximates $U(t)$ as a product of local exponentials, can be employed. Although so called ``beyond Trotter" methods, such as more general product formulas \cite{YOSHIDA1990262,BARTHEL2020168165,Ostmeyer2022}, quantum signal processing techniques \cite{Low17} \change{or Krylov subspace inspired algorithms \cite{Bharti_2021}}, have been investigated, recent studies on Trotter error scaling \cite{Childs2021, schubert2023trotter} have proven Trotter sequences to still be competitive.

Nonetheless, in specific applications, Trotter sequences fail to deliver the optimal solution for a product formula with fixed gate count \cite{jones2019optimising}. Amongst other optimization strategies, variational quantum algorithms have been put forward to find optimized (system-specific) gate sequences for Hamiltonian simulation \cite{Barison2021efficientquantum, Benedetti2021, C_rstoiu_2020, gibbs2021longtime, Tepaske2023, zhao2022}.

Although variational quantum algorithms have recently piqued significant interest as possible applications on near term quantum hardware \cite{Cerezo2021}, they hide intrinsic difficulties that make their implementation on near term devices very challenging. Vanishing gradients (Barren plateaus) \cite{McClean_2018, Cerezo_2021, Wang_2021}, sample overhead for efficient measurements \cite{Huang_2020, Zhenyu2020} and effects of noise on cost function evaluations \cite{Stilck_Fran_a_2021} are only some of the hurdles to overcome. Thus, key to a successful near term quantum algorithm is developing strategies to use quantum hardware in the crucial solution step, but otherwise as little as possible.

A drastic, however promising, approach to solve this is not to do quantum optimization at all, but rather push the whole optimization loop into classical pre-processing, which requires a cost function that is efficiently computable by a classical machine. This idea has been demonstrated to be applicable in specific cases, such as translational invariant models \cite{Mansuroglu_2023} in which parameters found on small systems can be reused in large systems, as well as systems with efficient tensor network representations \cite{Streif_2020, Keever2022, Cervero2023, kotil2022riemannian} to calculate error measures classically.

Here, we present a cost function that is classically efficiently computable and use it for classical optimization of quantum circuits that implement time evolution in pre-processing. We show that this cost function yields a faithful upper bound on the two-norm error of the ansatz unitary, if one restricts the single time step to be small, an assumption that is common to all product formulae \cite{YOSHIDA1990262,BARTHEL2020168165,Ostmeyer2022,Childs2021}. Our cost function is agnostic of the considered initial state and relies on a perturbative expansion of a rigorous upper bound that can be transformed into a polynomial in the parameters. The number of terms that need to be calculated scales with the number of non-commuting terms $H_j$ in the Hamiltonian. While this number is generally dependent on the locality and interaction distance in $H$, it is bounded by a power of the qubit number for models with only two-body interactions. There are thus no intrinsic limitations on the Hamiltonian for the use of our method. We further give a perturbative and classically computable error estimate that allows us to predict error scalings in system size and simulation time if the optimized gate sequence is repeated $K$ times. 

We illustrate the improvements enabled by our approach via exact numerical simulations for an XY model with random nearest neighboring interactions. Our results (cf. Fig.~\ref{fig:max_time}) show that the maximally reachable simulation times $T$ for a given target accuracy $\epsilon$ and a fixed gate budget $G=\mathcal{O}(K)$ are more than 10 times longer than for a Trotter sequence of the same order for cases with $K \gtrsim 20$ and $\epsilon=10^{-3}$. Moreover, the variational sequence even reaches slightly larger simulation times than second order Trotter and is competitive with fourth order Trotter, whereas these higher order formulas provide less data points for fixed gate budget.

The paper is structured as follows. We begin by introducing the second order perturbative distance of a variational sequence from the correct time evolution operator in section \ref{sec:pert} and show that it bounds the algorithm error up to third order corrections in the parameters. We further show that the perturbative distance always has a non-vanishing gradient at the point in the parameter space that corresponds to the Trotter sequence yielding a general guarantee for the existence of an optimal solution away from Trotter. In section \ref{sec:scaling}, we compute the scaling of the error in the extrapolation of found optimal gate sequence in system size and simulation time. 
% This supports previous findings for the special cases of translationally invariant systems \cite{Mansuroglu_2023}. 
Finally, we demonstrate our findings and the advantages they enable in numerical experiments on the XY model in section \ref{sec:num} before we close with a summary and outlook.

\section{Perturbative Cost Function}
\label{sec:pert}
We are interested in a variational ansatz for quantum simulation $U(t) = \exp{\left( -i t H \right)}$. It has been demonstrated before \cite{jones2019optimising, Mansuroglu_2023} that Trotter sequences can be optimized in a model specific way using the ansatz
\begin{align}
	U_{var} = \prod_{r=1}^R \prod_{j=1}^M e^{-i \theta_{r,j} H_j}.
	\label{eq:ansatz}
\end{align}
The inner product factorizes all $M$ Hamiltonian terms indexed with $j$ and the outer product adds the freedom to use in total $R$ layers in which the parameters can be different. \change{While the choice of ansatz is not necessarily restricted to Eq.~\eqref{eq:ansatz}, it has the advantage of containing the Trotter solution.} If we use the parameters $\theta_{r,j} = \frac{c_j t}{R}$, Eq.~\eqref{eq:ansatz} reduces to a Trotter decomposition with Trotter number $R$.

Finding the optimal solution for an operator approximation is hard, since error measures for the time evolution operator, such as
\begin{align}
	\epsilon_{var} = \norm{U(t) - U_{var}},
	\label{eq:error}
\end{align}
where $\norm{.}$ is a yet unspecified, unitary invariant norm, require exponential resources, in general. In this work, we present a classically efficient optimization strategy that yields good results in the relevant regime of small time steps $\frac{t}{R}$, where power expansions can be applied. To this end, we introduce an approximation to an upper bound of the error measure $\epsilon_{var}$ for the time evolution operator in equation (\ref{eq:error}), which is provided by a quantity that we call perturbative distance.
\begin{definition}
	Let $H = \sum_{j}^M c_j H_j$ be a time-independent Hamiltonian. Define the perturbative distance
	\begin{align}
		C(\theta) = \norm{
			i \sum_{j} \xi_j(\theta) H_j - \sum_{j>j'} \chi_{j,j'}(\theta) \comm{H_j}{H_{j'}} 
		}
		\label{eq:pert_dist}
	\end{align}
	with the coefficients
	\begin{align}
		\xi_j(\theta) &= \sum_r \theta_{r,j} - t c_j \\
		\chi_{j,j'}(\theta) &= \frac{1}{2} \left[ \sum_r \left( \theta_{r,j} \theta_{r,j'} \right) + \sum_{r>r'} \left( \theta_{r,j} \theta_{r',j'} - \theta_{r,j'} \theta_{r',j} \right) \right].
		\label{eq:chi_coeff}
	\end{align}
\end{definition}
The perturbative distance $C(\theta)$ represents an approximation to an upper bound of the error measure $\epsilon_{var}$ for the time evolution operator. We make this statement precise in the following proposition.

% \begin{align}
	%     \epsilon_{var} = \norm{U(t) - U_{var}},
	%     \label{eq:error}
	% \end{align}
% where $\norm{.}$ is a yet unspecified, unitary invariant norm.

\begin{theorem}
	Let $H = \sum_{j}^M c_j H_j$ be a time-independent Hamiltonian and $C(\theta)$ the perturbative distance defined in Eq.~\eqref{eq:pert_dist}. Then the error $\epsilon_{var}$ defined in Eq.~\eqref{eq:error} is bounded up to third order corrections by
	\begin{align}
		\epsilon_{var} \leq C(\theta) + \mathcal{O}\left(\theta^3\right)
	\end{align}
	\begin{proof}
		We can write the difference operator in Eq.~\eqref{eq:error} in the form
		\begin{align}
			e^{-itH} - e^Z,
		\end{align}
		where $Z$ is determined by merging $U_{var}$ into one exponential via the Baker-Campbell-Hausdorff (BCH) lemma. We show in Lemma \ref{lem:BCH_trunc} (see Appendix \ref{app:BCH_trunc} for a proof) that up to third order
		\begin{align}
			Z = &i \sum_{r,j} \theta_{r,j} H_j - \frac{1}{2} \sum_{(r,j)>(r',j')} \theta_{r,j} \theta_{r',j'} \comm{H_j}{H_{j'}} \nonumber \\
			&+ \mathcal{O}(\theta^3).
		\end{align}
		We assumed a specific ordering of the multi-index $(r,j)$ for the sake of being precise. The statement can be very well applied to arbitrary orderings of the exponentials in Eq.~\eqref{eq:ansatz}. We finally make use of a well-known estimate (see Lemma \ref{lem:bound} in Appendix \ref{app:bound} for a proof)
		\begin{align}
			\norm{e^{iA} - e^{iB}} \leq \norm{A - B},
		\end{align}
		that leaves us with the proposition.
	\end{proof}
\end{theorem}
To make $C(\theta)$ an approximate upper bound to the error $\epsilon_{var}$, we need to assume $\frac{t}{R} \norm{H} \ll 1$ and bound the parameters $\theta$ to be at the same order of magnitude. This assumption yields a time-dependent bound on $R$, which is not surprising. In general, Lieb-Robinson bounds on the entanglement that is generated by time evolution gives a relation between $t$ and $R$ \cite{Haah_2021}. If one manages to minimize $C(\theta)$ within the perturbative regime, the leading term in the upper bound will be of order $\mathcal{O}(\theta^3)$, yielding a new error estimate that we will analyze in section \ref{sec:scaling}.

Note that first order Trotter sequences can be represented by Eq.~\eqref{eq:ansatz} setting the parameters $\theta_{r,j} = \frac{t c_j}{R}$. The Trotter solution has the advantage that it makes the linear contribution in $\theta$, respectively $\frac{t}{R}$, vanish. It is however not the only solution that achieves this. In fact, we can reduce the number of parameters by $M$, if we constrain the linear terms by $\xi_j = 0 \, \forall j$. Let us thus fix the last layer ($r=R$) by setting
\begin{align}
	\theta_{R, j} = tc_j - \sum_{r \neq R} \theta_{r,j},
	\label{eq:cond}
\end{align}
so that $(R-1)M$ parameters remain. \change{Note that Eq.~\eqref{eq:cond} is a choice that turns out to be helpful as long as the operators that contribute in linear order are all orthogonal to the operators in the second order terms, which we assume in the following. Otherwise first and second order terms could be used to cancel each other.}

We are still free to choose a specific norm. While in finite dimensions, all norms are equivalent, choosing a norm to describe the error $\epsilon$, implies a notion of how the algorithm performs for specific initial states. If we were to choose the operator norm $\norm{.}_\infty$, for instance, $\epsilon$ can be connected to an infidelity for the worst case initial state, while a 2-norm $\norm{.}_2$ yields an average case error. Later, $C(\theta)$ shall be used as a cost function in variational optimization, so it should be continuously differentiable in the parameters $\theta_{r,j}$. While this rules out the operator norm, all Schatten $p$-norms can be applied. If we consider the 2-norm, we see that the computation of $C(\theta)$ is classically efficient. In this case
\begin{align}
	C(\theta)^2 = - 2^{-n} \sum_{j>j'} \sum_{k>k'} \chi_{j,j'} \chi_{k,k'} \Tr\left( \comm{H_j}{H_{j'}} \comm{H_k}{H_{k'}} \right),
	\label{eq:cost}
\end{align}
where we included a normalization factor $2^{-n}$ that prevents $C(\theta)$ from growing in the Hilbert space dimension. Typically, the commutators $\comm{H_j}{H_{j'}}$ follow a simple algebraic structure (for instance if all $H_j$ are Pauli strings). The factors $\Tr\left( \comm{H_j}{H_{j'}} \comm{H_k}{H_{k'}} \right)$ can thus efficiently be computed for the model of interest, typically even by hand, leaving $C(\theta)$ to be a fourth order polynomial in $\theta$.

The Trotter solution, while becoming exact in the limit $R\to\infty$, can always be beaten for finite $R \geq 3$, as we will show in the following.
\begin{theorem}
	\label{theo:opt_pot}
	Let $H= \sum_j c_j H_j$ a Hamiltonian and $U_{var}$ as defined in Eq.~\eqref{eq:ansatz} with $R\geq 3$. The Trotter solution $\theta^T_{r,j} = \frac{t c_j}{R}$ describes a (local) minimum of $C(\theta)$ if and only if the Trotter formula is exact, i.e. $C(\theta^T)=0 \, \forall t$.
	\begin{proof}[Proof (Sketch)]
		For the forward direction "$\Rightarrow$", we explicitly calculate the gradient at the Trotter solution $\nabla C\left(\theta^T\right)$. Setting this to zero, gives a relation
		\begin{align}
			\sum_{j>l} \sum_{k>k'} \mathcal{A}_{j,k,k',l} = \frac{2r - R}{2(r-1) - R}  \sum_{j<l} \sum_{k>k'} \mathcal{A}_{j,k,k',l} \quad \forall r,l
			\label{eq:C_grad}
		\end{align}
		where $\mathcal{A}_{j,k,k',l}$ is a tensor independent of $r$ and $R$. For $R \geq 3$, we show that this forces the left hand side of Eq.~\eqref{eq:C_grad} to vanish making the Trotter sequence exact. For details, see Appendix \ref{app:proof_pot}. The reverse direction "$\Leftarrow$" is trivial.
	\end{proof}
\end{theorem}

\section{Error Scaling}
\label{sec:scaling}
Since the proposed classical optimization scheme can only make valid statements within the perturbative regime $\theta \ll \frac{1}{\norm{H}}$, longer simulation times either require large $R$ or repetitions of a small time step. A similar extrapolation can be done in system size. This section studies the scaling of the perturbative distance along these extrapolations. In the case that the second order perturbative cost function is optimzed to values that become negligible to higher orders, the error scaling is determined by a third order term $E(\theta)$ for which
\begin{align}
	\epsilon_{var} \leq C(\theta) + E(\theta) + \mathcal{O} \left( \theta^4 \right)
	\label{eq:error_term}
\end{align}
holds. First, we consider the extrapolation in system size in the presence of translation symmetries.
\begin{theorem}
	\label{theo:trans_inv}
	Let $H = \sum_{a=1}^A \sum_{j=1}^n c_a H^{(j)}_{a}$ be a $p$-local, translation invariant Hamiltonian, where the index $a$ labels different interaction terms and the dependence on the qubit number, indicated by the superscript $(j)$, only denotes on which qubit $H_a^{(j)}$ has support. Let further $U_{var}$ be as defined in Eq.~\eqref{eq:ansatz}. 
	
	If the parameters are chosen to be translation invariant, i.e. $\theta_{r,j} = \theta_{r,a}$, then the cost value for a system of $n$ qubits on a $D$-dimensional geometry is bounded by the cost value $C_{\rm unit}(\theta)$ of a system of $(2p-1)^D$ qubits and the error is bounded by $E_{\rm unit}(\theta)$ of a system of $(3p-2)^D$ qubits
	\begin{align}
		C(\theta) \leq n C_{\rm unit}(\theta) \qquad \text{and} \qquad E(\theta) \leq n E_{\rm unit}(\theta).
	\end{align}
	Further, if all the commutators $\comm{H_j}{H_{j'}}$ are orthogonal with respect to the Hilbert-Schmidt product, cost and error read
	\begin{align}
		C(\theta) = \sqrt{n} C_{\rm unit}(\theta) \quad \text{and} \quad E(\theta) = \sqrt{n} E_{\rm unit}(\theta).
	\end{align}
\end{theorem}
See Appendix \ref{app:trans_inv} for a proof. Remarkably, the calculation of $C(\theta)$ on translation invariant systems is reduced to a problem on $(2p-1)^D$ qubits only. If the optimal parameters are used for a larger system, the cost function will scale as $\mathcal{O}(n)$. In fact, a system that is small enough can be exactly diagonalized on a classical computer and an exact measure of error, such as $\norm{U(0,t) - U_{var}}$ can be calculated. A successful extrapolation of optimal parameters for translation invariant systems to larger system sizes has been found in \cite{Mansuroglu_2023} and also shows a linear scaling of the squared cost function in system size $n$, even beyond the perturbative regime. \change{While we restrict the discussion here to translation symmetries, the behavior of other scaling transformation such as clusterings or pooling inspired transformations can be analogously checked.}

Optimizing on longer simulation times still within the perturbative regime comes with an increase in $R$ and therefore with the number of parameters making a successful optimization hard. Instead, long simulation times can be reached by repeating a pre-trained single time step fixing the number of parameters to be $RM$. The scaling of the cost function under this extrapolation in simulation time is studied in the following.
\begin{theorem}
	\label{theo:long_time}
	Let $U_{var}(\theta)$ be an ansatz for variational Hamiltonian simulation as defined in Eq.~\eqref{eq:ansatz} with the cost value $C(\theta)$. The ansatz
	\begin{align}
		U_{var}^K\left( \theta^K \right) = \left( \prod_{r=1}^R \prod_{j=1}^M e^{-i \theta_{r,j} H_j} \right)^K.
		\label{eq:Uvar_long}
	\end{align}
	simulates time $T = Kt$ using $KR$ layers and $K$ copies of the $RM$ parameters. The corresponding cost value reads 
	\begin{align}
		C\left( \theta^K \right) = K C(\theta).
	\end{align}
	\begin{proof}[Proof idea]
		The extrapolated ansatz is equivalent to the ansatz before by mapping $R \to KR$, $t \to Kt$ and identifying $\theta_{r,j} = \theta_{r \% R,j}$, where $\%$ denotes the modulo operation. Plugging this into the definition of $\chi$ from Eq.~\eqref{eq:chi_coeff} yields a linear factor $K$. See Appendix \ref{app:long_times} for details.
	\end{proof}
\end{theorem}
In total, the cost function scales as $\mathcal{O}(K)$ when extrapolating in time. Note that the leading order Trotter error scales the same under the above mapping
\begin{align}
	\norm{U(0,t) - U_{Trotter}} = \mathcal{O}\left( \frac{t^2}{R} \right) \to \mathcal{O}\left( K \frac{t^2}{R} \right).
\end{align}
The error $E(\theta)$ of the repeated sequence does not admit a common $K$-dependent factor that can be pulled out, but rather introduces extra terms, $\Lambda$, that make the error scaling worse than for Trotter sequences (see appendix \ref{app:long_times} for details). The extra terms will, in general, contribute to an $\mathcal{O}(K^2)$ scaling,
\begin{align}
	E(\theta^K) \leq K E(\theta) + \frac{K(K-1)}{2} \Lambda.
\end{align}
Although this result means that usually the improvement with respect to Trotter will eventually melt away at long times, our numerical experiments show that there are significant improvements for the time scales that are relevant in order to reach quantum advantage in digital quantum simulation. 

\section{Numerical Experiment}
\label{sec:num}
To explore the accuracy improvement that can be enabled with our approach in specific examples, we study the performance of the derived cost function by optimizing a variational sequence with the second order perturbative distance defined in Eq. (\ref{eq:pert_dist}) and comparing it to the exact value for the difference of time evolutions $\epsilon_{var} = \norm{U(t) - U_{var}}$ and $ \epsilon_{Trotter} = \norm{U(t) - U_{Trotter}}$ as well as to the leading order error term $E(\theta)$ (cf. Eq.~\eqref{eq:error_term}) that is efficiently computable classically. For this comparison, we study the ratios 
\begin{align}
	R_\epsilon = \frac{\epsilon_{Trotter}}{\epsilon_{Var}} \qquad R_E = \frac{\epsilon_{Trotter}}{E(\theta)}
	\label{eq:ratios}
\end{align}
of perturbative ($R_E$) and exact ($R_\epsilon$) error measures comparing Trotter and variational sequences. In Eq. (\ref{eq:ratios}), $C(\theta^T)$ denotes the perturbative distance for choosing the variational parameters equal to the Trotter choice $\theta_{r,j} = \frac{t c_j}{R}$. 

From the bound in Eq.~\eqref{eq:error_term}, we can also infer an approximate lower bound for the exact error ratio
\begin{align}
	R_\epsilon \geq \frac{\epsilon_{Trotter}}{C(\theta) + E(\theta)} + \mathcal{O}(\theta^4) \approx R_E,
	\label{eq:error_bound}
\end{align}
where we assumed that $C(\theta) \ll E(\theta)$ after a successful optimization of the parameters $\theta$.

As an example, we consider an XY model with a fully connected interaction graph that is described by the Hamiltonian \change{(up to unitary equivalence)}
\begin{align}
	H_{XY} = - \sum_{\mu>\nu} \left( J_{\mu\nu}^{(y)} Y_\mu Y_\nu + J_{\mu\nu}^{(z)} Z_\mu Z_\nu \right) + \sum_\mu h_\mu X_\mu.
	\label{eq:Ham_XY}
\end{align}
The sum over $\mu,\nu$ counts every pair of qubits and the symmetric matrices $J^{(y,z)}_{\mu\nu}$ encode the interaction strengths between qubits $\mu$ and $\nu$. In our numerical simulations, we normalize the Hamiltonian to $\norm{H} = \sqrt{n}$ to fix the time scale.

To calculate $C(\theta)$ for the XY model, we define a mapping between the single index $j \in \{1, ..., M\}$ used in the above derivations and the Hamiltonians in Eq.~\eqref{eq:Ham_XY} labelled by the qubit indices $\mu \in \{1, ..., n\}$ and the interaction type. The non-vanishing commutator terms that appear in $C(\theta)$ are of the form $\comm{Z_\mu Z_\nu}{X_\mu}, \comm{Y_\mu Y_\nu}{X_\mu}$ or $\comm{Z_\mu Z_\nu}{Y_\nu Y_\sigma}$. This allows the reduction of $C(\theta)$ to $\mathcal{O}\left(n^3\right)$ non-vanishing terms in the fully connected case and to $\mathcal{O}(n)$ non-vanishing terms in a nearest-neighbor XY model. See Appendix \ref{app:XY} for an explicit form of the perturbative distance $C(\theta)$. \change{Note that we are still free to turn off specific interactions from the fully connected graph to study lattice geometries of any dimension. We will now turn to a two-dimensional geometry for which there is no analytical solution using free fermion mapping, such as in one dimension.}

In Fig. \ref{fig:IsingXY3x3}, we plot the cost values $C, E$ and $\epsilon$, as well as the ratios $R_\epsilon$ and $R_E$ for different times $t$. The optimizer manages to find small values of $C(\theta)$ for all times. For small enough times $t < \frac{R}{\norm{H}} = 1$, the perturbative distance also incorporates a faithful indicator for the exact norm in a sense that the optimal parameters also yield a significant decrease in the exact error $\epsilon_{Var}$ of the variational sequence. Also, the error estimate $E$ and the ratio $R_E$ yield very good approximations for the true error ratio $R_\epsilon$ as predicted in Eq.~\eqref{eq:error_bound}. The numerics hence indicate that the approximate upper bound derived in Eq.~\eqref{eq:error_term} is tight.
\begin{figure*}[t]
	\centering
	\includegraphics[scale=0.7]{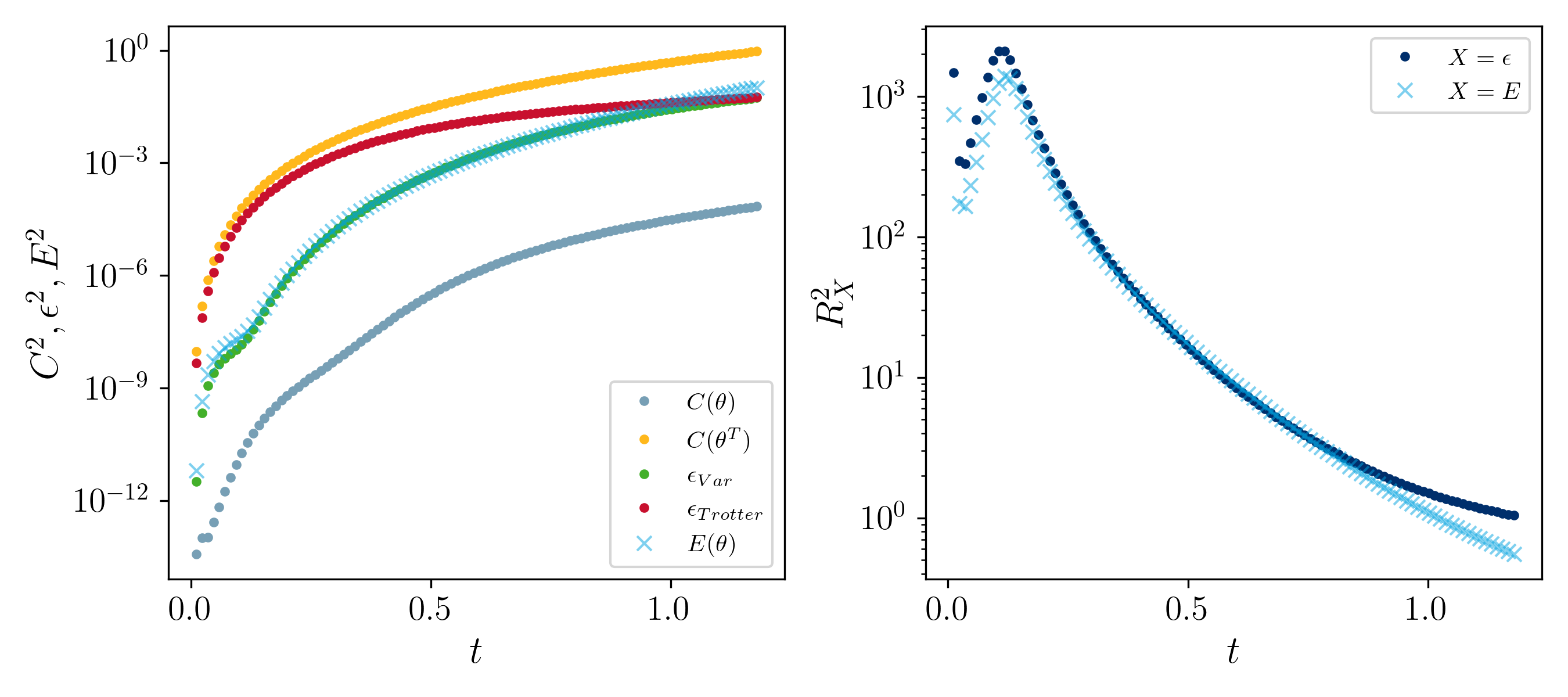}
	\caption{\textbf{Perturbative ($C, E$) and exact ($\epsilon$) error measures (left) and exact error ratio $R_\epsilon$ and the ratio of error estimates $R_E$ (right)} as defined in Eq.~\eqref{eq:ratios} evaluated on a nearest neighbor XY model on a $3\times3$ quadratic lattice and random interaction strength $J^{(y)}_{\mu,\nu}$ centered around 0.5, $J^{(z)}_{\mu\nu}$ centered around 1 and $h_\mu=0.25$.}
	\label{fig:IsingXY3x3}
\end{figure*}

In Fig. \ref{fig:IsingXY3x3_repeat}, the optimal parameters found before for a fixed time step $\sqrt{2n} t\in \{0.05, 0.5, 1.5\}$ are being repeated up to $K=10$ times to simulate longer times. As predicted in Proposition \ref{theo:long_time}, the square of the perturbative distance grows with a factor $K^2$. For total simulation times $\sqrt{2n} Kt < 1.5$, this scaling is also reflected in the exact cost $\epsilon$. Above this critical time, higher order contributions become non-negligible, which explains why the scaling $\epsilon$ deviates from $\mathcal{O}\left( K^2 \right)$. Remarkably, the error ratio between Trotter and variational sequences stays constant also far beyond this critical time until the improvement of the variational sequence eventually starts melting away at $t>1$.
\begin{figure*}[t]
	\centering
	\includegraphics[scale=0.59]{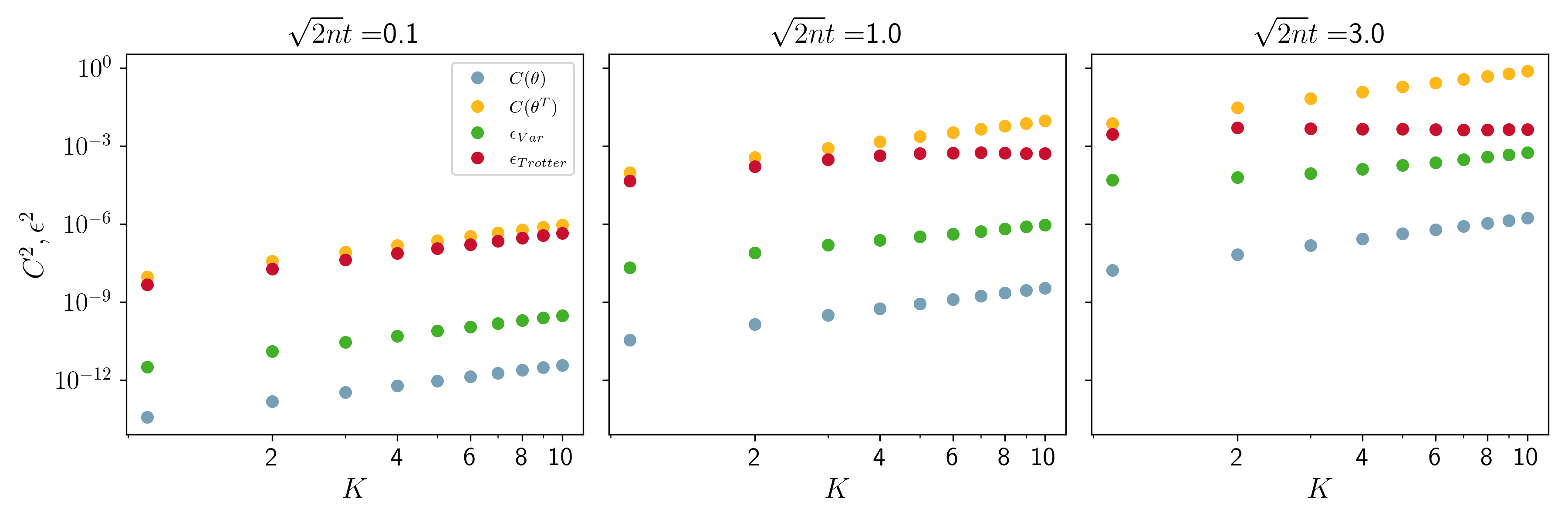}
	\caption{\textbf{Pertubartive cost values $C$ compared with exact norm cost function $\epsilon$ } for a repeated sequence of variational and Trotter solution evaluated on a nearest neighbor XY model on a $3\times3$ quadratic lattice and random interaction strength $J^{(y)}_{\mu,\nu}$ centered around 0.5, $J^{(z)}_{\mu\nu}$ centered around 1 and $h_\mu=0.25$. A single time step always corresponds to $R=3$ layers and the single step time increases from left to right $\sqrt{2n} t\in \{0.05, 0.5, 1.5\}$. From the log-log plot, one can read off $C(\theta^K)^2 = \mathcal{O}(K^2)$ as predicted by Proposition \ref{theo:long_time}.}
	\label{fig:IsingXY3x3_repeat}
\end{figure*}

Let us take this extrapolation to a practical case in which we have a gate budget $G = \mathcal{O}(KRn)$, that will scale with the number of qubits $n$ and the number of circuit repetitions $K$, and a maximal target accuracy $\epsilon$. A natural question from a practitioner's point of view is to ask, what is the maximally reachable simulation time $T=Kt$ with these fixed resources. We study this question in Fig.~\ref{fig:max_time} where we plot $T_{V}$ for variational and Trotter sequences of order $q \in \{1,2,4\}$ picking the maximal time step $t$ that corresponds to an error $\epsilon$ just below the threshold. While the variational sequence reaches larger simulation times than first and second order Trotter sequences, its maximal simulation time is comparable to the fourth order Trotter formula in the $\epsilon = 10^{-2}$ case and in the $\epsilon = 10^{-3}$ case for circuits with $K<25$. The simulation time of the variational sequence is then surpassed by that of a $q=4$ Trotter formula in the case $\epsilon = 10^{-3}$ at depths $K>25$. However, since higher order Trotter formulas require more gates for a single time step, the resolution of the time evolution for a given gate number also becomes more sparse. For NISQ applications, which are typically $K<20$, there might be only a single or no fourth order data point in the implementable range, depending on the system size. We also plot the ratio of simulation times $T_T/T_V$ for first order sequences to show that our method can reach simulation times that are more than 10 times longer than for a Trotter sequence of the same order for $\epsilon = 10^{-3}$ and $K \gtrsim 20$.

We have also computed the values of the  the ratios $R_C$ and $R_\epsilon$ for different times $t$ for a Transverse Field Ising model in Appendix \ref{app:Ising}. Also for this model, our approach leads to a dramatic improvement. This model however has a particularly simple commutator algebra that is routed in it's low connectivity. Therefore many of the two-qubit gates in a Trotter sequence commute with each other and, as a consequence, a first order Trotter sequence of the Ising model is unitarily equivalent to a second order Trotter sequence. Improving upon this sequence would require a third order perturbative distance analogous to our discussion above.

\section{Conclusion}
In this work, we presented a classically efficient way to estimate error measures for variational time evolution in the regime $\frac{t}{R}, \theta \ll \frac{1}{\norm{H}}$ and proposed an optimization routine for variational Hamiltonian simulation that is executed completely in classical pre-processing. The method is applicable to arbitrary $p$-local Hamiltonians and only requires computing commutators of individual terms in the Hamiltonian.

We show that -- analogous to a Trotter decomposition -- the linear contribution to the error in the parameters $\theta$ can be forced to exactly vanish. The remaining parameters can then be found to minimize the second order error terms for which the existence of a better solution than Trotter is always guaranteed within the perturbative regime. A generalization to higher orders is straight-forward and remains only a matter of diligence. The Trotter-inspired error scalings in system size and extrapolations in simulation time are conserved also when deviating from the Trotter solution. Our findings are backed up with numerical experiments on the XY model with random interaction strengths. 

For future work, it would be interesting to study other physically motivated models like non-local fermionic encodings. Also, an extension to time-dependent Hamiltonians would be interesting, although the extrapolation to larger times is not straight-forward for this case. Another possibility to simplify the computation of the perturbative distance is to change the exact error measure that is to be approximated. This can be done by fixing the initial state of the quantum simulation or to consider alternative error measures, such as the conservation of statistical moments of the Hamiltonian, as recently proposed in \cite{zhao2022}.

\begin{acknowledgments}
    This work received support from the German Federal Ministry of Education and Research via the funding program quantum technologies - from basic research to the market under contract number 13N16067 “EQUAHUMO”. It is also part of the Munich Quantum Valley, which is supported by the Bavarian state government with funds from the Hightech Agenda Bayern Plus.
\end{acknowledgments}

\newpage
%\twocolumngrid
%\bibliographystyle{plain}
\bibliography{literature}

\begin{thebibliography}{31}%
\makeatletter
\providecommand \@ifxundefined [1]{%
 \@ifx{#1\undefined}
}%
\providecommand \@ifnum [1]{%
 \ifnum #1\expandafter \@firstoftwo
 \else \expandafter \@secondoftwo
 \fi
}%
\providecommand \@ifx [1]{%
 \ifx #1\expandafter \@firstoftwo
 \else \expandafter \@secondoftwo
 \fi
}%
\providecommand \natexlab [1]{#1}%
\providecommand \enquote  [1]{``#1''}%
\providecommand \bibnamefont  [1]{#1}%
\providecommand \bibfnamefont [1]{#1}%
\providecommand \citenamefont [1]{#1}%
\providecommand \href@noop [0]{\@secondoftwo}%
\providecommand \href [0]{\begingroup \@sanitize@url \@href}%
\providecommand \@href[1]{\@@startlink{#1}\@@href}%
\providecommand \@@href[1]{\endgroup#1\@@endlink}%
\providecommand \@sanitize@url [0]{\catcode `\\12\catcode `\$12\catcode
  `\&12\catcode `\#12\catcode `\^12\catcode `\_12\catcode `\%12\relax}%
\providecommand \@@startlink[1]{}%
\providecommand \@@endlink[0]{}%
\providecommand \url  [0]{\begingroup\@sanitize@url \@url }%
\providecommand \@url [1]{\endgroup\@href {#1}{\urlprefix }}%
\providecommand \urlprefix  [0]{URL }%
\providecommand \Eprint [0]{\href }%
\providecommand \doibase [0]{https://doi.org/}%
\providecommand \selectlanguage [0]{\@gobble}%
\providecommand \bibinfo  [0]{\@secondoftwo}%
\providecommand \bibfield  [0]{\@secondoftwo}%
\providecommand \translation [1]{[#1]}%
\providecommand \BibitemOpen [0]{}%
\providecommand \bibitemStop [0]{}%
\providecommand \bibitemNoStop [0]{.\EOS\space}%
\providecommand \EOS [0]{\spacefactor3000\relax}%
\providecommand \BibitemShut  [1]{\csname bibitem#1\endcsname}%
\let\auto@bib@innerbib\@empty
%</preamble>
\bibitem [{\citenamefont {Smith}\ \emph {et~al.}(2019)\citenamefont {Smith},
  \citenamefont {Kim}, \citenamefont {Pollmann},\ and\ \citenamefont
  {Knolle}}]{Smith_2019}%
  \BibitemOpen
  \bibfield  {author} {\bibinfo {author} {\bibfnamefont {A.}~\bibnamefont
  {Smith}}, \bibinfo {author} {\bibfnamefont {M.~S.}\ \bibnamefont {Kim}},
  \bibinfo {author} {\bibfnamefont {F.}~\bibnamefont {Pollmann}},\ and\
  \bibinfo {author} {\bibfnamefont {J.}~\bibnamefont {Knolle}},\ }\bibfield
  {title} {\bibinfo {title} {Simulating quantum many-body dynamics on a current
  digital quantum computer},\ }\bibfield  {journal} {\bibinfo  {journal} {npj
  Quantum Information}\ }\textbf {\bibinfo {volume} {5}},\ \href
  {https://doi.org/10.1038/s41534-019-0217-0} {10.1038/s41534-019-0217-0}
  (\bibinfo {year} {2019})\BibitemShut {NoStop}%
\bibitem [{\citenamefont {Arute}\ \emph {et~al.}(2019)\citenamefont {Arute},
  \citenamefont {Arya}, \citenamefont {Babbush}, \citenamefont {Bacon},
  \citenamefont {Bardin}, \citenamefont {Barends}, \citenamefont {Biswas},
  \citenamefont {Boixo}, \citenamefont {Brandao}, \citenamefont {Buell},
  \citenamefont {Burkett}, \citenamefont {Chen}, \citenamefont {Chen},
  \citenamefont {Chiaro}, \citenamefont {Collins}, \citenamefont {Courtney},
  \citenamefont {Dunsworth}, \citenamefont {Farhi}, \citenamefont {Foxen},
  \citenamefont {Fowler}, \citenamefont {Gidney}, \citenamefont {Giustina},
  \citenamefont {Graff}, \citenamefont {Guerin}, \citenamefont {Habegger},
  \citenamefont {Harrigan}, \citenamefont {Hartmann}, \citenamefont {Ho},
  \citenamefont {Hoffmann}, \citenamefont {Huang}, \citenamefont {Humble},
  \citenamefont {Isakov}, \citenamefont {Jeffrey}, \citenamefont {Jiang},
  \citenamefont {Kafri}, \citenamefont {Kechedzhi}, \citenamefont {Kelly},
  \citenamefont {Klimov}, \citenamefont {Knysh}, \citenamefont {Korotkov},
  \citenamefont {Kostritsa}, \citenamefont {Landhuis}, \citenamefont
  {Lindmark}, \citenamefont {Lucero}, \citenamefont {Lyakh}, \citenamefont
  {Mandr{\`a}}, \citenamefont {McClean}, \citenamefont {McEwen}, \citenamefont
  {Megrant}, \citenamefont {Mi}, \citenamefont {Michielsen}, \citenamefont
  {Mohseni}, \citenamefont {Mutus}, \citenamefont {Naaman}, \citenamefont
  {Neeley}, \citenamefont {Neill}, \citenamefont {Niu}, \citenamefont {Ostby},
  \citenamefont {Petukhov}, \citenamefont {Platt}, \citenamefont {Quintana},
  \citenamefont {Rieffel}, \citenamefont {Roushan}, \citenamefont {Rubin},
  \citenamefont {Sank}, \citenamefont {Satzinger}, \citenamefont {Smelyanskiy},
  \citenamefont {Sung}, \citenamefont {Trevithick}, \citenamefont
  {Vainsencher}, \citenamefont {Villalonga}, \citenamefont {White},
  \citenamefont {Yao}, \citenamefont {Yeh}, \citenamefont {Zalcman},
  \citenamefont {Neven},\ and\ \citenamefont {Martinis}}]{Arute2019}%
  \BibitemOpen
  \bibfield  {author} {\bibinfo {author} {\bibfnamefont {F.}~\bibnamefont
  {Arute}}, \bibinfo {author} {\bibfnamefont {K.}~\bibnamefont {Arya}},
  \bibinfo {author} {\bibfnamefont {R.}~\bibnamefont {Babbush}}, \bibinfo
  {author} {\bibfnamefont {D.}~\bibnamefont {Bacon}}, \bibinfo {author}
  {\bibfnamefont {J.~C.}\ \bibnamefont {Bardin}}, \bibinfo {author}
  {\bibfnamefont {R.}~\bibnamefont {Barends}}, \bibinfo {author} {\bibfnamefont
  {R.}~\bibnamefont {Biswas}}, \bibinfo {author} {\bibfnamefont
  {S.}~\bibnamefont {Boixo}}, \bibinfo {author} {\bibfnamefont {F.~G. S.~L.}\
  \bibnamefont {Brandao}}, \bibinfo {author} {\bibfnamefont {D.~A.}\
  \bibnamefont {Buell}}, \bibinfo {author} {\bibfnamefont {B.}~\bibnamefont
  {Burkett}}, \bibinfo {author} {\bibfnamefont {Y.}~\bibnamefont {Chen}},
  \bibinfo {author} {\bibfnamefont {Z.}~\bibnamefont {Chen}}, \bibinfo {author}
  {\bibfnamefont {B.}~\bibnamefont {Chiaro}}, \bibinfo {author} {\bibfnamefont
  {R.}~\bibnamefont {Collins}}, \bibinfo {author} {\bibfnamefont
  {W.}~\bibnamefont {Courtney}}, \bibinfo {author} {\bibfnamefont
  {A.}~\bibnamefont {Dunsworth}}, \bibinfo {author} {\bibfnamefont
  {E.}~\bibnamefont {Farhi}}, \bibinfo {author} {\bibfnamefont
  {B.}~\bibnamefont {Foxen}}, \bibinfo {author} {\bibfnamefont
  {A.}~\bibnamefont {Fowler}}, \bibinfo {author} {\bibfnamefont
  {C.}~\bibnamefont {Gidney}}, \bibinfo {author} {\bibfnamefont
  {M.}~\bibnamefont {Giustina}}, \bibinfo {author} {\bibfnamefont
  {R.}~\bibnamefont {Graff}}, \bibinfo {author} {\bibfnamefont
  {K.}~\bibnamefont {Guerin}}, \bibinfo {author} {\bibfnamefont
  {S.}~\bibnamefont {Habegger}}, \bibinfo {author} {\bibfnamefont {M.~P.}\
  \bibnamefont {Harrigan}}, \bibinfo {author} {\bibfnamefont {M.~J.}\
  \bibnamefont {Hartmann}}, \bibinfo {author} {\bibfnamefont {A.}~\bibnamefont
  {Ho}}, \bibinfo {author} {\bibfnamefont {M.}~\bibnamefont {Hoffmann}},
  \bibinfo {author} {\bibfnamefont {T.}~\bibnamefont {Huang}}, \bibinfo
  {author} {\bibfnamefont {T.~S.}\ \bibnamefont {Humble}}, \bibinfo {author}
  {\bibfnamefont {S.~V.}\ \bibnamefont {Isakov}}, \bibinfo {author}
  {\bibfnamefont {E.}~\bibnamefont {Jeffrey}}, \bibinfo {author} {\bibfnamefont
  {Z.}~\bibnamefont {Jiang}}, \bibinfo {author} {\bibfnamefont
  {D.}~\bibnamefont {Kafri}}, \bibinfo {author} {\bibfnamefont
  {K.}~\bibnamefont {Kechedzhi}}, \bibinfo {author} {\bibfnamefont
  {J.}~\bibnamefont {Kelly}}, \bibinfo {author} {\bibfnamefont {P.~V.}\
  \bibnamefont {Klimov}}, \bibinfo {author} {\bibfnamefont {S.}~\bibnamefont
  {Knysh}}, \bibinfo {author} {\bibfnamefont {A.}~\bibnamefont {Korotkov}},
  \bibinfo {author} {\bibfnamefont {F.}~\bibnamefont {Kostritsa}}, \bibinfo
  {author} {\bibfnamefont {D.}~\bibnamefont {Landhuis}}, \bibinfo {author}
  {\bibfnamefont {M.}~\bibnamefont {Lindmark}}, \bibinfo {author}
  {\bibfnamefont {E.}~\bibnamefont {Lucero}}, \bibinfo {author} {\bibfnamefont
  {D.}~\bibnamefont {Lyakh}}, \bibinfo {author} {\bibfnamefont
  {S.}~\bibnamefont {Mandr{\`a}}}, \bibinfo {author} {\bibfnamefont {J.~R.}\
  \bibnamefont {McClean}}, \bibinfo {author} {\bibfnamefont {M.}~\bibnamefont
  {McEwen}}, \bibinfo {author} {\bibfnamefont {A.}~\bibnamefont {Megrant}},
  \bibinfo {author} {\bibfnamefont {X.}~\bibnamefont {Mi}}, \bibinfo {author}
  {\bibfnamefont {K.}~\bibnamefont {Michielsen}}, \bibinfo {author}
  {\bibfnamefont {M.}~\bibnamefont {Mohseni}}, \bibinfo {author} {\bibfnamefont
  {J.}~\bibnamefont {Mutus}}, \bibinfo {author} {\bibfnamefont
  {O.}~\bibnamefont {Naaman}}, \bibinfo {author} {\bibfnamefont
  {M.}~\bibnamefont {Neeley}}, \bibinfo {author} {\bibfnamefont
  {C.}~\bibnamefont {Neill}}, \bibinfo {author} {\bibfnamefont {M.~Y.}\
  \bibnamefont {Niu}}, \bibinfo {author} {\bibfnamefont {E.}~\bibnamefont
  {Ostby}}, \bibinfo {author} {\bibfnamefont {A.}~\bibnamefont {Petukhov}},
  \bibinfo {author} {\bibfnamefont {J.~C.}\ \bibnamefont {Platt}}, \bibinfo
  {author} {\bibfnamefont {C.}~\bibnamefont {Quintana}}, \bibinfo {author}
  {\bibfnamefont {E.~G.}\ \bibnamefont {Rieffel}}, \bibinfo {author}
  {\bibfnamefont {P.}~\bibnamefont {Roushan}}, \bibinfo {author} {\bibfnamefont
  {N.~C.}\ \bibnamefont {Rubin}}, \bibinfo {author} {\bibfnamefont
  {D.}~\bibnamefont {Sank}}, \bibinfo {author} {\bibfnamefont {K.~J.}\
  \bibnamefont {Satzinger}}, \bibinfo {author} {\bibfnamefont {V.}~\bibnamefont
  {Smelyanskiy}}, \bibinfo {author} {\bibfnamefont {K.~J.}\ \bibnamefont
  {Sung}}, \bibinfo {author} {\bibfnamefont {M.~D.}\ \bibnamefont
  {Trevithick}}, \bibinfo {author} {\bibfnamefont {A.}~\bibnamefont
  {Vainsencher}}, \bibinfo {author} {\bibfnamefont {B.}~\bibnamefont
  {Villalonga}}, \bibinfo {author} {\bibfnamefont {T.}~\bibnamefont {White}},
  \bibinfo {author} {\bibfnamefont {Z.~J.}\ \bibnamefont {Yao}}, \bibinfo
  {author} {\bibfnamefont {P.}~\bibnamefont {Yeh}}, \bibinfo {author}
  {\bibfnamefont {A.}~\bibnamefont {Zalcman}}, \bibinfo {author} {\bibfnamefont
  {H.}~\bibnamefont {Neven}},\ and\ \bibinfo {author} {\bibfnamefont {J.~M.}\
  \bibnamefont {Martinis}},\ }\bibfield  {title} {\bibinfo {title} {Quantum
  supremacy using a programmable superconducting processor},\ }\href
  {https://doi.org/10.1038/s41586-019-1666-5} {\bibfield  {journal} {\bibinfo
  {journal} {Nature}\ }\textbf {\bibinfo {volume} {574}},\ \bibinfo {pages}
  {505} (\bibinfo {year} {2019})}\BibitemShut {NoStop}%
\bibitem [{\citenamefont {Suzuki}(1985)}]{Suzuki_85}%
  \BibitemOpen
  \bibfield  {author} {\bibinfo {author} {\bibfnamefont {M.}~\bibnamefont
  {Suzuki}},\ }\bibfield  {title} {\bibinfo {title} {Decomposition formulas of
  exponential operators and lie exponentials with some applications to quantum
  mechanics and statistical physics},\ }\href
  {https://doi.org/10.1063/1.526596} {\bibfield  {journal} {\bibinfo  {journal}
  {Journal of Mathematical Physics}\ }\textbf {\bibinfo {volume} {26}},\
  \bibinfo {pages} {601} (\bibinfo {year} {1985})},\ \Eprint
  {https://arxiv.org/abs/https://doi.org/10.1063/1.526596}
  {https://doi.org/10.1063/1.526596} \BibitemShut {NoStop}%
\bibitem [{\citenamefont {Yoshida}(1990)}]{YOSHIDA1990262}%
  \BibitemOpen
  \bibfield  {author} {\bibinfo {author} {\bibfnamefont {H.}~\bibnamefont
  {Yoshida}},\ }\bibfield  {title} {\bibinfo {title} {Construction of higher
  order symplectic integrators},\ }\href
  {https://doi.org/https://doi.org/10.1016/0375-9601(90)90092-3} {\bibfield
  {journal} {\bibinfo  {journal} {Physics Letters A}\ }\textbf {\bibinfo
  {volume} {150}},\ \bibinfo {pages} {262} (\bibinfo {year}
  {1990})}\BibitemShut {NoStop}%
\bibitem [{\citenamefont {Barthel}\ and\ \citenamefont
  {Zhang}(2020)}]{BARTHEL2020168165}%
  \BibitemOpen
  \bibfield  {author} {\bibinfo {author} {\bibfnamefont {T.}~\bibnamefont
  {Barthel}}\ and\ \bibinfo {author} {\bibfnamefont {Y.}~\bibnamefont
  {Zhang}},\ }\bibfield  {title} {\bibinfo {title} {Optimized
  lie–trotter–suzuki decompositions for two and three non-commuting
  terms},\ }\href {https://doi.org/https://doi.org/10.1016/j.aop.2020.168165}
  {\bibfield  {journal} {\bibinfo  {journal} {Annals of Physics}\ }\textbf
  {\bibinfo {volume} {418}},\ \bibinfo {pages} {168165} (\bibinfo {year}
  {2020})}\BibitemShut {NoStop}%
\bibitem [{\citenamefont {Ostmeyer}(2023)}]{Ostmeyer2022}%
  \BibitemOpen
  \bibfield  {author} {\bibinfo {author} {\bibfnamefont {J.}~\bibnamefont
  {Ostmeyer}},\ }\bibfield  {title} {\bibinfo {title} {Optimised trotter
  decompositions for classical and quantum computing},\ }\href
  {https://doi.org/10.1088/1751-8121/acde7a} {\bibfield  {journal} {\bibinfo
  {journal} {Journal of Physics A: Mathematical and Theoretical}\ }\textbf
  {\bibinfo {volume} {56}},\ \bibinfo {pages} {285303} (\bibinfo {year}
  {2023})}\BibitemShut {NoStop}%
\bibitem [{\citenamefont {Low}\ and\ \citenamefont {Chuang}(2017)}]{Low17}%
  \BibitemOpen
  \bibfield  {author} {\bibinfo {author} {\bibfnamefont {G.~H.}\ \bibnamefont
  {Low}}\ and\ \bibinfo {author} {\bibfnamefont {I.~L.}\ \bibnamefont
  {Chuang}},\ }\bibfield  {title} {\bibinfo {title} {Optimal hamiltonian
  simulation by quantum signal processing},\ }\href
  {https://doi.org/10.1103/PhysRevLett.118.010501} {\bibfield  {journal}
  {\bibinfo  {journal} {Phys. Rev. Lett.}\ }\textbf {\bibinfo {volume} {118}},\
  \bibinfo {pages} {010501} (\bibinfo {year} {2017})}\BibitemShut {NoStop}%
\bibitem [{\citenamefont {Bharti}\ and\ \citenamefont
  {Haug}(2021)}]{Bharti_2021}%
  \BibitemOpen
  \bibfield  {author} {\bibinfo {author} {\bibfnamefont {K.}~\bibnamefont
  {Bharti}}\ and\ \bibinfo {author} {\bibfnamefont {T.}~\bibnamefont {Haug}},\
  }\bibfield  {title} {\bibinfo {title} {Quantum-assisted simulator},\
  }\bibfield  {journal} {\bibinfo  {journal} {Physical Review A}\ }\textbf
  {\bibinfo {volume} {104}},\ \href
  {https://doi.org/10.1103/physreva.104.042418} {10.1103/physreva.104.042418}
  (\bibinfo {year} {2021})\BibitemShut {NoStop}%
\bibitem [{\citenamefont {Childs}\ \emph {et~al.}(2021)\citenamefont {Childs},
  \citenamefont {Su}, \citenamefont {Tran}, \citenamefont {Wiebe},\ and\
  \citenamefont {Zhu}}]{Childs2021}%
  \BibitemOpen
  \bibfield  {author} {\bibinfo {author} {\bibfnamefont {A.~M.}\ \bibnamefont
  {Childs}}, \bibinfo {author} {\bibfnamefont {Y.}~\bibnamefont {Su}}, \bibinfo
  {author} {\bibfnamefont {M.~C.}\ \bibnamefont {Tran}}, \bibinfo {author}
  {\bibfnamefont {N.}~\bibnamefont {Wiebe}},\ and\ \bibinfo {author}
  {\bibfnamefont {S.}~\bibnamefont {Zhu}},\ }\bibfield  {title} {\bibinfo
  {title} {Theory of trotter error with commutator scaling},\ }\href
  {https://doi.org/10.1103/PhysRevX.11.011020} {\bibfield  {journal} {\bibinfo
  {journal} {Phys. Rev. X}\ }\textbf {\bibinfo {volume} {11}},\ \bibinfo
  {pages} {011020} (\bibinfo {year} {2021})}\BibitemShut {NoStop}%
\bibitem [{\citenamefont {Schubert}\ and\ \citenamefont
  {Mendl}(2023)}]{schubert2023trotter}%
  \BibitemOpen
  \bibfield  {author} {\bibinfo {author} {\bibfnamefont {A.}~\bibnamefont
  {Schubert}}\ and\ \bibinfo {author} {\bibfnamefont {C.~B.}\ \bibnamefont
  {Mendl}},\ }\href@noop {} {\bibinfo {title} {Trotter error with commutator
  scaling for the fermi-hubbard model}} (\bibinfo {year} {2023}),\ \Eprint
  {https://arxiv.org/abs/2306.10603} {arXiv:2306.10603 [quant-ph]} \BibitemShut
  {NoStop}%
\bibitem [{\citenamefont {Jones}\ \emph {et~al.}(2019)\citenamefont {Jones},
  \citenamefont {O'Brien}, \citenamefont {White}, \citenamefont {Campbell},\
  and\ \citenamefont {Clark}}]{jones2019optimising}%
  \BibitemOpen
  \bibfield  {author} {\bibinfo {author} {\bibfnamefont {B.~D.~M.}\
  \bibnamefont {Jones}}, \bibinfo {author} {\bibfnamefont {G.~O.}\ \bibnamefont
  {O'Brien}}, \bibinfo {author} {\bibfnamefont {D.~R.}\ \bibnamefont {White}},
  \bibinfo {author} {\bibfnamefont {E.~T.}\ \bibnamefont {Campbell}},\ and\
  \bibinfo {author} {\bibfnamefont {J.~A.}\ \bibnamefont {Clark}},\ }\href@noop
  {} {\bibinfo {title} {Optimising trotter-suzuki decompositions for quantum
  simulation using evolutionary strategies}} (\bibinfo {year} {2019}),\ \Eprint
  {https://arxiv.org/abs/1904.01336} {arXiv:1904.01336 [cs.NE]} \BibitemShut
  {NoStop}%
\bibitem [{\citenamefont {Barison}\ \emph {et~al.}(2021)\citenamefont
  {Barison}, \citenamefont {Vicentini},\ and\ \citenamefont
  {Carleo}}]{Barison2021efficientquantum}%
  \BibitemOpen
  \bibfield  {author} {\bibinfo {author} {\bibfnamefont {S.}~\bibnamefont
  {Barison}}, \bibinfo {author} {\bibfnamefont {F.}~\bibnamefont {Vicentini}},\
  and\ \bibinfo {author} {\bibfnamefont {G.}~\bibnamefont {Carleo}},\
  }\bibfield  {title} {\bibinfo {title} {An efficient quantum algorithm for the
  time evolution of parameterized circuits},\ }\href
  {https://doi.org/10.22331/q-2021-07-28-512} {\bibfield  {journal} {\bibinfo
  {journal} {{Quantum}}\ }\textbf {\bibinfo {volume} {5}},\ \bibinfo {pages}
  {512} (\bibinfo {year} {2021})}\BibitemShut {NoStop}%
\bibitem [{\citenamefont {Benedetti}\ \emph {et~al.}(2021)\citenamefont
  {Benedetti}, \citenamefont {Fiorentini},\ and\ \citenamefont
  {Lubasch}}]{Benedetti2021}%
  \BibitemOpen
  \bibfield  {author} {\bibinfo {author} {\bibfnamefont {M.}~\bibnamefont
  {Benedetti}}, \bibinfo {author} {\bibfnamefont {M.}~\bibnamefont
  {Fiorentini}},\ and\ \bibinfo {author} {\bibfnamefont {M.}~\bibnamefont
  {Lubasch}},\ }\bibfield  {title} {\bibinfo {title} {Hardware-efficient
  variational quantum algorithms for time evolution},\ }\href
  {https://doi.org/10.1103/PhysRevResearch.3.033083} {\bibfield  {journal}
  {\bibinfo  {journal} {Phys. Rev. Res.}\ }\textbf {\bibinfo {volume} {3}},\
  \bibinfo {pages} {033083} (\bibinfo {year} {2021})}\BibitemShut {NoStop}%
\bibitem [{\citenamefont {C{\^{\i}}rstoiu}\ \emph {et~al.}(2020)\citenamefont
  {C{\^{\i}}rstoiu}, \citenamefont {Holmes}, \citenamefont {Iosue},
  \citenamefont {Cincio}, \citenamefont {Coles},\ and\ \citenamefont
  {Sornborger}}]{C_rstoiu_2020}%
  \BibitemOpen
  \bibfield  {author} {\bibinfo {author} {\bibfnamefont {C.}~\bibnamefont
  {C{\^{\i}}rstoiu}}, \bibinfo {author} {\bibfnamefont {Z.}~\bibnamefont
  {Holmes}}, \bibinfo {author} {\bibfnamefont {J.}~\bibnamefont {Iosue}},
  \bibinfo {author} {\bibfnamefont {L.}~\bibnamefont {Cincio}}, \bibinfo
  {author} {\bibfnamefont {P.~J.}\ \bibnamefont {Coles}},\ and\ \bibinfo
  {author} {\bibfnamefont {A.}~\bibnamefont {Sornborger}},\ }\bibfield  {title}
  {\bibinfo {title} {Variational fast forwarding for quantum simulation beyond
  the coherence time},\ }\bibfield  {journal} {\bibinfo  {journal} {npj Quantum
  Information}\ }\textbf {\bibinfo {volume} {6}},\ \href
  {https://doi.org/10.1038/s41534-020-00302-0} {10.1038/s41534-020-00302-0}
  (\bibinfo {year} {2020})\BibitemShut {NoStop}%
\bibitem [{\citenamefont {Gibbs}\ \emph {et~al.}(2022)\citenamefont {Gibbs},
  \citenamefont {Gili}, \citenamefont {Holmes}, \citenamefont {Commeau},
  \citenamefont {Arrasmith}, \citenamefont {Cincio}, \citenamefont {Coles},\
  and\ \citenamefont {Sornborger}}]{gibbs2021longtime}%
  \BibitemOpen
  \bibfield  {author} {\bibinfo {author} {\bibfnamefont {J.}~\bibnamefont
  {Gibbs}}, \bibinfo {author} {\bibfnamefont {K.}~\bibnamefont {Gili}},
  \bibinfo {author} {\bibfnamefont {Z.}~\bibnamefont {Holmes}}, \bibinfo
  {author} {\bibfnamefont {B.}~\bibnamefont {Commeau}}, \bibinfo {author}
  {\bibfnamefont {A.}~\bibnamefont {Arrasmith}}, \bibinfo {author}
  {\bibfnamefont {L.}~\bibnamefont {Cincio}}, \bibinfo {author} {\bibfnamefont
  {P.~J.}\ \bibnamefont {Coles}},\ and\ \bibinfo {author} {\bibfnamefont
  {A.}~\bibnamefont {Sornborger}},\ }\bibfield  {title} {\bibinfo {title}
  {Long-time simulations for fixed input states on quantum hardware},\ }\href
  {https://doi.org/10.1038/s41534-022-00625-0} {\bibfield  {journal} {\bibinfo
  {journal} {npj Quantum Information}\ }\textbf {\bibinfo {volume} {8}},\
  \bibinfo {pages} {135} (\bibinfo {year} {2022})}\BibitemShut {NoStop}%
\bibitem [{\citenamefont {Tepaske}\ \emph {et~al.}(2023)\citenamefont
  {Tepaske}, \citenamefont {Hahn},\ and\ \citenamefont {Luitz}}]{Tepaske2023}%
  \BibitemOpen
  \bibfield  {author} {\bibinfo {author} {\bibfnamefont {M.~S.~J.}\
  \bibnamefont {Tepaske}}, \bibinfo {author} {\bibfnamefont {D.}~\bibnamefont
  {Hahn}},\ and\ \bibinfo {author} {\bibfnamefont {D.~J.}\ \bibnamefont
  {Luitz}},\ }\bibfield  {title} {\bibinfo {title} {{Optimal compression of
  quantum many-body time evolution operators into brickwall circuits}},\ }\href
  {https://doi.org/10.21468/SciPostPhys.14.4.073} {\bibfield  {journal}
  {\bibinfo  {journal} {SciPost Phys.}\ }\textbf {\bibinfo {volume} {14}},\
  \bibinfo {pages} {073} (\bibinfo {year} {2023})}\BibitemShut {NoStop}%
\bibitem [{\citenamefont {Zhao}\ \emph {et~al.}(2023)\citenamefont {Zhao},
  \citenamefont {Bukov}, \citenamefont {Heyl},\ and\ \citenamefont
  {Moessner}}]{zhao2022}%
  \BibitemOpen
  \bibfield  {author} {\bibinfo {author} {\bibfnamefont {H.}~\bibnamefont
  {Zhao}}, \bibinfo {author} {\bibfnamefont {M.}~\bibnamefont {Bukov}},
  \bibinfo {author} {\bibfnamefont {M.}~\bibnamefont {Heyl}},\ and\ \bibinfo
  {author} {\bibfnamefont {R.}~\bibnamefont {Moessner}},\ }\bibfield  {title}
  {\bibinfo {title} {Making trotterization adaptive and energy-self-correcting
  for nisq devices and beyond},\ }\href
  {https://doi.org/10.1103/PRXQuantum.4.030319} {\bibfield  {journal} {\bibinfo
   {journal} {PRX Quantum}\ }\textbf {\bibinfo {volume} {4}},\ \bibinfo {pages}
  {030319} (\bibinfo {year} {2023})}\BibitemShut {NoStop}%
\bibitem [{\citenamefont {Cerezo}\ \emph
  {et~al.}(2021{\natexlab{a}})\citenamefont {Cerezo}, \citenamefont
  {Arrasmith}, \citenamefont {Babbush}, \citenamefont {Benjamin}, \citenamefont
  {Endo}, \citenamefont {Fujii}, \citenamefont {McClean}, \citenamefont
  {Mitarai}, \citenamefont {Yuan}, \citenamefont {Cincio},\ and\ \citenamefont
  {Coles}}]{Cerezo2021}%
  \BibitemOpen
  \bibfield  {author} {\bibinfo {author} {\bibfnamefont {M.}~\bibnamefont
  {Cerezo}}, \bibinfo {author} {\bibfnamefont {A.}~\bibnamefont {Arrasmith}},
  \bibinfo {author} {\bibfnamefont {R.}~\bibnamefont {Babbush}}, \bibinfo
  {author} {\bibfnamefont {S.~C.}\ \bibnamefont {Benjamin}}, \bibinfo {author}
  {\bibfnamefont {S.}~\bibnamefont {Endo}}, \bibinfo {author} {\bibfnamefont
  {K.}~\bibnamefont {Fujii}}, \bibinfo {author} {\bibfnamefont {J.~R.}\
  \bibnamefont {McClean}}, \bibinfo {author} {\bibfnamefont {K.}~\bibnamefont
  {Mitarai}}, \bibinfo {author} {\bibfnamefont {X.}~\bibnamefont {Yuan}},
  \bibinfo {author} {\bibfnamefont {L.}~\bibnamefont {Cincio}},\ and\ \bibinfo
  {author} {\bibfnamefont {P.~J.}\ \bibnamefont {Coles}},\ }\bibfield  {title}
  {\bibinfo {title} {Variational quantum algorithms},\ }\href
  {https://doi.org/10.1038/s42254-021-00348-9} {\bibfield  {journal} {\bibinfo
  {journal} {Nature Reviews Physics}\ }\textbf {\bibinfo {volume} {3}},\
  \bibinfo {pages} {625} (\bibinfo {year} {2021}{\natexlab{a}})}\BibitemShut
  {NoStop}%
\bibitem [{\citenamefont {McClean}\ \emph {et~al.}(2018)\citenamefont
  {McClean}, \citenamefont {Boixo}, \citenamefont {Smelyanskiy}, \citenamefont
  {Babbush},\ and\ \citenamefont {Neven}}]{McClean_2018}%
  \BibitemOpen
  \bibfield  {author} {\bibinfo {author} {\bibfnamefont {J.~R.}\ \bibnamefont
  {McClean}}, \bibinfo {author} {\bibfnamefont {S.}~\bibnamefont {Boixo}},
  \bibinfo {author} {\bibfnamefont {V.~N.}\ \bibnamefont {Smelyanskiy}},
  \bibinfo {author} {\bibfnamefont {R.}~\bibnamefont {Babbush}},\ and\ \bibinfo
  {author} {\bibfnamefont {H.}~\bibnamefont {Neven}},\ }\bibfield  {title}
  {\bibinfo {title} {Barren plateaus in quantum neural network training
  landscapes},\ }\bibfield  {journal} {\bibinfo  {journal} {Nature
  Communications}\ }\textbf {\bibinfo {volume} {9}},\ \href
  {https://doi.org/10.1038/s41467-018-07090-4} {10.1038/s41467-018-07090-4}
  (\bibinfo {year} {2018})\BibitemShut {NoStop}%
\bibitem [{\citenamefont {Cerezo}\ \emph
  {et~al.}(2021{\natexlab{b}})\citenamefont {Cerezo}, \citenamefont {Sone},
  \citenamefont {Volkoff}, \citenamefont {Cincio},\ and\ \citenamefont
  {Coles}}]{Cerezo_2021}%
  \BibitemOpen
  \bibfield  {author} {\bibinfo {author} {\bibfnamefont {M.}~\bibnamefont
  {Cerezo}}, \bibinfo {author} {\bibfnamefont {A.}~\bibnamefont {Sone}},
  \bibinfo {author} {\bibfnamefont {T.}~\bibnamefont {Volkoff}}, \bibinfo
  {author} {\bibfnamefont {L.}~\bibnamefont {Cincio}},\ and\ \bibinfo {author}
  {\bibfnamefont {P.~J.}\ \bibnamefont {Coles}},\ }\bibfield  {title} {\bibinfo
  {title} {Cost function dependent barren plateaus in shallow parametrized
  quantum circuits},\ }\bibfield  {journal} {\bibinfo  {journal} {Nature
  Communications}\ }\textbf {\bibinfo {volume} {12}},\ \href
  {https://doi.org/10.1038/s41467-021-21728-w} {10.1038/s41467-021-21728-w}
  (\bibinfo {year} {2021}{\natexlab{b}})\BibitemShut {NoStop}%
\bibitem [{\citenamefont {Wang}\ \emph {et~al.}(2021)\citenamefont {Wang},
  \citenamefont {Fontana}, \citenamefont {Cerezo}, \citenamefont {Sharma},
  \citenamefont {Sone}, \citenamefont {Cincio},\ and\ \citenamefont
  {Coles}}]{Wang_2021}%
  \BibitemOpen
  \bibfield  {author} {\bibinfo {author} {\bibfnamefont {S.}~\bibnamefont
  {Wang}}, \bibinfo {author} {\bibfnamefont {E.}~\bibnamefont {Fontana}},
  \bibinfo {author} {\bibfnamefont {M.}~\bibnamefont {Cerezo}}, \bibinfo
  {author} {\bibfnamefont {K.}~\bibnamefont {Sharma}}, \bibinfo {author}
  {\bibfnamefont {A.}~\bibnamefont {Sone}}, \bibinfo {author} {\bibfnamefont
  {L.}~\bibnamefont {Cincio}},\ and\ \bibinfo {author} {\bibfnamefont {P.~J.}\
  \bibnamefont {Coles}},\ }\bibfield  {title} {\bibinfo {title} {Noise-induced
  barren plateaus in variational quantum algorithms},\ }\bibfield  {journal}
  {\bibinfo  {journal} {Nature Communications}\ }\textbf {\bibinfo {volume}
  {12}},\ \href {https://doi.org/10.1038/s41467-021-27045-6}
  {10.1038/s41467-021-27045-6} (\bibinfo {year} {2021})\BibitemShut {NoStop}%
\bibitem [{\citenamefont {Huang}\ \emph {et~al.}(2020)\citenamefont {Huang},
  \citenamefont {Kueng},\ and\ \citenamefont {Preskill}}]{Huang_2020}%
  \BibitemOpen
  \bibfield  {author} {\bibinfo {author} {\bibfnamefont {H.-Y.}\ \bibnamefont
  {Huang}}, \bibinfo {author} {\bibfnamefont {R.}~\bibnamefont {Kueng}},\ and\
  \bibinfo {author} {\bibfnamefont {J.}~\bibnamefont {Preskill}},\ }\bibfield
  {title} {\bibinfo {title} {Predicting many properties of a quantum system
  from very few measurements},\ }\href
  {https://doi.org/10.1038/s41567-020-0932-7} {\bibfield  {journal} {\bibinfo
  {journal} {Nature Physics}\ }\textbf {\bibinfo {volume} {16}},\ \bibinfo
  {pages} {1050} (\bibinfo {year} {2020})}\BibitemShut {NoStop}%
\bibitem [{\citenamefont {Cai}(2020)}]{Zhenyu2020}%
  \BibitemOpen
  \bibfield  {author} {\bibinfo {author} {\bibfnamefont {Z.}~\bibnamefont
  {Cai}},\ }\bibfield  {title} {\bibinfo {title} {Resource estimation for
  quantum variational simulations of the hubbard model},\ }\href
  {https://doi.org/10.1103/PhysRevApplied.14.014059} {\bibfield  {journal}
  {\bibinfo  {journal} {Phys. Rev. Appl.}\ }\textbf {\bibinfo {volume} {14}},\
  \bibinfo {pages} {014059} (\bibinfo {year} {2020})}\BibitemShut {NoStop}%
\bibitem [{\citenamefont {Fran{\c{c}}a}\ and\ \citenamefont
  {Garc{\'{\i}}a-Patr{\'{o}}n}(2021)}]{Stilck_Fran_a_2021}%
  \BibitemOpen
  \bibfield  {author} {\bibinfo {author} {\bibfnamefont {D.~S.}\ \bibnamefont
  {Fran{\c{c}}a}}\ and\ \bibinfo {author} {\bibfnamefont {R.}~\bibnamefont
  {Garc{\'{\i}}a-Patr{\'{o}}n}},\ }\bibfield  {title} {\bibinfo {title}
  {Limitations of optimization algorithms on noisy quantum devices},\ }\href
  {https://doi.org/10.1038/s41567-021-01356-3} {\bibfield  {journal} {\bibinfo
  {journal} {Nature Physics}\ }\textbf {\bibinfo {volume} {17}},\ \bibinfo
  {pages} {1221} (\bibinfo {year} {2021})}\BibitemShut {NoStop}%
\bibitem [{\citenamefont {Mansuroglu}\ \emph {et~al.}(2023)\citenamefont
  {Mansuroglu}, \citenamefont {Eckstein}, \citenamefont {Nützel},
  \citenamefont {Wilkinson},\ and\ \citenamefont {Hartmann}}]{Mansuroglu_2023}%
  \BibitemOpen
  \bibfield  {author} {\bibinfo {author} {\bibfnamefont {R.}~\bibnamefont
  {Mansuroglu}}, \bibinfo {author} {\bibfnamefont {T.}~\bibnamefont
  {Eckstein}}, \bibinfo {author} {\bibfnamefont {L.}~\bibnamefont {Nützel}},
  \bibinfo {author} {\bibfnamefont {S.~A.}\ \bibnamefont {Wilkinson}},\ and\
  \bibinfo {author} {\bibfnamefont {M.~J.}\ \bibnamefont {Hartmann}},\
  }\bibfield  {title} {\bibinfo {title} {Variational hamiltonian simulation for
  translational invariant systems via classical pre-processing},\ }\href
  {https://doi.org/10.1088/2058-9565/acb1d0} {\bibfield  {journal} {\bibinfo
  {journal} {Quantum Science and Technology}\ }\textbf {\bibinfo {volume}
  {8}},\ \bibinfo {pages} {025006} (\bibinfo {year} {2023})}\BibitemShut
  {NoStop}%
\bibitem [{\citenamefont {Streif}\ and\ \citenamefont
  {Leib}(2020)}]{Streif_2020}%
  \BibitemOpen
  \bibfield  {author} {\bibinfo {author} {\bibfnamefont {M.}~\bibnamefont
  {Streif}}\ and\ \bibinfo {author} {\bibfnamefont {M.}~\bibnamefont {Leib}},\
  }\bibfield  {title} {\bibinfo {title} {Training the quantum approximate
  optimization algorithm without access to a quantum processing unit},\ }\href
  {https://doi.org/10.1088/2058-9565/ab8c2b} {\bibfield  {journal} {\bibinfo
  {journal} {Quantum Science and Technology}\ }\textbf {\bibinfo {volume}
  {5}},\ \bibinfo {pages} {034008} (\bibinfo {year} {2020})}\BibitemShut
  {NoStop}%
\bibitem [{\citenamefont {Keever}\ and\ \citenamefont
  {Lubasch}(2023)}]{Keever2022}%
  \BibitemOpen
  \bibfield  {author} {\bibinfo {author} {\bibfnamefont {C.~M.}\ \bibnamefont
  {Keever}}\ and\ \bibinfo {author} {\bibfnamefont {M.}~\bibnamefont
  {Lubasch}},\ }\bibfield  {title} {\bibinfo {title} {Classically optimized
  hamiltonian simulation},\ }\bibfield  {journal} {\bibinfo  {journal}
  {Physical Review Research}\ }\textbf {\bibinfo {volume} {5}},\ \href
  {https://doi.org/10.1103/physrevresearch.5.023146}
  {10.1103/physrevresearch.5.023146} (\bibinfo {year} {2023})\BibitemShut
  {NoStop}%
\bibitem [{\citenamefont {{Cervero Mart{\'\i}n}}\ \emph
  {et~al.}(2023)\citenamefont {{Cervero Mart{\'\i}n}}, \citenamefont
  {{Plekhanov}},\ and\ \citenamefont {{Lubasch}}}]{Cervero2023}%
  \BibitemOpen
  \bibfield  {author} {\bibinfo {author} {\bibfnamefont {E.}~\bibnamefont
  {{Cervero Mart{\'\i}n}}}, \bibinfo {author} {\bibfnamefont {K.}~\bibnamefont
  {{Plekhanov}}},\ and\ \bibinfo {author} {\bibfnamefont {M.}~\bibnamefont
  {{Lubasch}}},\ }\bibfield  {title} {\bibinfo {title} {{Barren plateaus in
  quantum tensor network optimization}},\ }\href
  {https://doi.org/10.22331/q-2023-04-13-974} {\bibfield  {journal} {\bibinfo
  {journal} {Quantum}\ }\textbf {\bibinfo {volume} {7}},\ \bibinfo {pages}
  {974} (\bibinfo {year} {2023})},\ \Eprint {https://arxiv.org/abs/2209.00292}
  {arXiv:2209.00292 [quant-ph]} \BibitemShut {NoStop}%
\bibitem [{\citenamefont {Kotil}\ \emph {et~al.}(2022)\citenamefont {Kotil},
  \citenamefont {Banerjee}, \citenamefont {Huang},\ and\ \citenamefont
  {Mendl}}]{kotil2022riemannian}%
  \BibitemOpen
  \bibfield  {author} {\bibinfo {author} {\bibfnamefont {A.}~\bibnamefont
  {Kotil}}, \bibinfo {author} {\bibfnamefont {R.}~\bibnamefont {Banerjee}},
  \bibinfo {author} {\bibfnamefont {Q.}~\bibnamefont {Huang}},\ and\ \bibinfo
  {author} {\bibfnamefont {C.~B.}\ \bibnamefont {Mendl}},\ }\href@noop {}
  {\bibinfo {title} {Riemannian quantum circuit optimization for hamiltonian
  simulation}} (\bibinfo {year} {2022}),\ \Eprint
  {https://arxiv.org/abs/2212.07556} {arXiv:2212.07556 [quant-ph]} \BibitemShut
  {NoStop}%
\bibitem [{\citenamefont {Haah}\ \emph {et~al.}(2021)\citenamefont {Haah},
  \citenamefont {Hastings}, \citenamefont {Kothari},\ and\ \citenamefont
  {Low}}]{Haah_2021}%
  \BibitemOpen
  \bibfield  {author} {\bibinfo {author} {\bibfnamefont {J.}~\bibnamefont
  {Haah}}, \bibinfo {author} {\bibfnamefont {M.~B.}\ \bibnamefont {Hastings}},
  \bibinfo {author} {\bibfnamefont {R.}~\bibnamefont {Kothari}},\ and\ \bibinfo
  {author} {\bibfnamefont {G.~H.}\ \bibnamefont {Low}},\ }\bibfield  {title}
  {\bibinfo {title} {Quantum algorithm for simulating real time evolution of
  lattice hamiltonians},\ }\href {https://doi.org/10.1137/18m1231511}
  {\bibfield  {journal} {\bibinfo  {journal} {{SIAM} Journal on Computing}\ ,\
  \bibinfo {pages} {FOCS18}} (\bibinfo {year} {2021})}\BibitemShut {NoStop}%
\bibitem [{\citenamefont {Chakraborty}\ \emph {et~al.}(2019)\citenamefont
  {Chakraborty}, \citenamefont {Gily{\'e}n},\ and\ \citenamefont
  {Jeffery}}]{chakraborty2019}%
  \BibitemOpen
  \bibfield  {author} {\bibinfo {author} {\bibfnamefont {S.}~\bibnamefont
  {Chakraborty}}, \bibinfo {author} {\bibfnamefont {A.}~\bibnamefont
  {Gily{\'e}n}},\ and\ \bibinfo {author} {\bibfnamefont {S.}~\bibnamefont
  {Jeffery}},\ }\bibfield  {title} {\bibinfo {title} {{The Power of
  Block-Encoded Matrix Powers: Improved Regression Techniques via Faster
  Hamiltonian Simulation}},\ }in\ \href
  {https://doi.org/10.4230/LIPIcs.ICALP.2019.33} {\emph {\bibinfo {booktitle}
  {46th International Colloquium on Automata, Languages, and Programming (ICALP
  2019)}}},\ \bibinfo {series} {Leibniz International Proceedings in
  Informatics (LIPIcs)}, Vol.\ \bibinfo {volume} {132},\ \bibinfo {editor}
  {edited by\ \bibinfo {editor} {\bibfnamefont {C.}~\bibnamefont {Baier}},
  \bibinfo {editor} {\bibfnamefont {I.}~\bibnamefont {Chatzigiannakis}},
  \bibinfo {editor} {\bibfnamefont {P.}~\bibnamefont {Flocchini}},\ and\
  \bibinfo {editor} {\bibfnamefont {S.}~\bibnamefont {Leonardi}}}\ (\bibinfo
  {publisher} {Schloss Dagstuhl--Leibniz-Zentrum fuer Informatik},\ \bibinfo
  {address} {Dagstuhl, Germany},\ \bibinfo {year} {2019})\ pp.\ \bibinfo
  {pages} {33:1--33:14}\BibitemShut {NoStop}%
\end{thebibliography}%

\onecolumngrid
\appendix
\newpage
\section{Derivation of the Cost Function}
\subsection{An upper bound for the error}
\label{app:bound}
To derive an upper bound for the error of a variational gate sequence for time evolution, we start by reviewing a general upper bound for differences of unitaries shown in \cite{chakraborty2019}.
\begin{lemma}
	\label{lem:bound}
	Let $A, B$ be Hermitian operators and $\norm{.}$ any norm that admits unitary invariance, then
	\begin{align}
		\norm{e^{iA} - e^{iB}} \leq \norm{A - B}
	\end{align}
	\begin{proof}
		We begin by writing the difference operator $e^{i A} - e^{i B}$ in integral form
		\begin{align}
			e^{i A} - e^{i B} &= \int_0^1 dx \frac{d}{dx} \left[ e^{i \left(B + x(A - B) \right)} \right] = \int_0^1 dx \int_0^1 dy \left[ e^{i y \left(B + x(A - B) \right)} (A - B) \times e^{i (1-y) \left(B + x(A - B) \right)} \right].
			\label{eq:diff_op}
		\end{align}
		In the first step, we used the fundamental theorem of calculus and in the second step, we used the identity for the derivative of the exponential map
		\begin{align}
			\frac{d}{dx} e^{F(x)} = \int_0^1 dy \, e^{y F(x)} F'(x) e^{(1-y) F(x)}
		\end{align}
		for any operator valued function $F(x)$ and its derivative $F'(x)$. If we consider the norm of the operator in Eq.~\eqref{eq:diff_op}, we can use the triangle inequality and its unitary invariance to get to the result
		\begin{align}
			\norm{e^{i A} - e^{i B}} \leq &\int_0^1 dx \int_0^1 dy \norm{ e^{i y \left(B + x(A - B) \right)} (A - B) e^{i (1-y) \left(B + x(A - B) \right)}} \nonumber \\
			=&\int_0^1 dx \int_0^1 dy \norm{A - B} = \norm{A - B}.
		\end{align}
	\end{proof}
\end{lemma}
Let $H=\sum_{j=1}^M c_j H_j$ be Hermitian. We are interested in a time evolution operator built from gates of the form $e^{-i \theta H_j}$. For this choice, the variational gate sequence will have the form of Eq.~\eqref{eq:ansatz}. We can use the Baker-Campbell-Hausdorff formula to find an anti-Hermitian operator $Z$ that combines all exponentials in Eq.~\eqref{eq:ansatz} into one. Using Lemma \ref{lem:bound}, we then get a general bound for the difference operator
\begin{align}
	\norm{U(0,t) - U_{var}} = \norm{e^{-i t H} - e^Z} \leq \norm{itH + Z}.
	\label{eq:gen_bound}
\end{align}
$Z$ contains infinitely many terms in general. To transform Eq.~\eqref{eq:gen_bound} into something that is classically computable, we will truncate $Z$ in a perturbative regime.

\subsection{Truncation of BCH-formula}
\label{app:BCH_trunc}
In the limit of small time steps $\frac{t}{R} \to 0$, a Trotter formula becomes exact. We want to keep $\frac{t}{R}$ small but finite in the following. We further assume the parameters $\theta_{r,j}$ to be of the same order as $\frac{t}{R}$. Consider the following supplementary Lemma
\begin{lemma}
	\label{lem:BCH_trunc}
	Let $A_j, j\in \mathfrak{J}$ for some index set $\mathfrak{J}$, be a list of bounded operators and denote the maximal operator norm $\max_j \norm{A_j}_\infty =: a$. It holds
	\begin{align}
		&\prod_{j \in \mathfrak{J}}^\leftarrow e^{A_j} = e^B \\
		B &= \sum_{j\in \mathfrak{I}} A_j + \frac{1}{2} \sum_{j>k} \comm{A_j}{A_k} + \frac{1}{12} \sum_{j \neq k} \comm{A_j}{\comm{A_j}{A_{k}}} + \frac{1}{6} \sum_{j>k>l} \left( \comm{A_j}{\comm{A_k}{A_{l}}} + \comm{A_l}{\comm{A_k}{A_{j}}} \right) + \mathcal{O}\left(a^4\right),
		\label{eq:multi_BCH}
	\end{align}
	where we assumed an implicit ordering on the index set $\mathfrak{J}$, and the arrow above the product implies that the exponential corresponding to the smallest index is being multiplied from the right.
	\begin{proof}
		We proceed by induction over the size of the index set $\mathfrak{J}$. For $|\mathfrak{J}|=1$, the statement becomes trivial and for $|\mathfrak{J}|=2$ equivalent to the Baker-Campbell-Hausdorff formula. We hence assume that the statement holds for an index set $\mathfrak{J}$ of size $I-1$. Let us rename $Z$ from above to $Z_{I-1}$ to emphasize the size of the index set. If we add an $I^\text{th}$ index, we get
		\begin{align}
			e^{A_I} \prod_{j \in \mathfrak{J}}^\leftarrow e^{A_j} = e^{A_I} e^{Z_{I-1}} = e^{Z_I},
		\end{align}
		where we get $Z_I$ from the Baker-Campbell-Hausdorff lemma
		\begin{align}
			Z_I &= A_I + Z_{I-1} + \frac{1}{2} \comm{A_I}{Z_{I-1}} + \frac{1}{12} \comm{A_I - Z_{I-1}}{\comm{A_I}{Z_{I-1}}} + \mathcal{O}\left( a^4 \right).
			\label{eq:Z_I}
		\end{align}
		Let us recollect all orders in $a$. For the first order we get $A_I + \sum_{i=1}^{I-1} A_i = \sum_{i=1}^{I} A_i$. For the second order term, we gather two terms from Eq.~\eqref{eq:Z_I}
		\begin{align} 
			\frac{1}{2} \sum_{i>j}^{I-1} \comm{A_i}{A_j} + \frac{1}{2} \comm{A_I}{\sum_{i=1}^{I-1} A_i} + \mathcal{O}(a^3) = \frac{1}{2} \sum_{i>j}^{I} \comm{A_i}{A_j} + \mathcal{O}(a^3)
		\end{align}
		Finally, the third order term has three contributions from Eq.~\eqref{eq:Z_I}
		\begin{align}
			&\frac{1}{12} \sum_{j \neq k}^{I-1} \comm{A_j}{\comm{A_j}{A_k}} + \frac{1}{6} \sum_{j>k>l}^{I-1} \left( \comm{A_j}{\comm{A_k}{A_{l}}} + \comm{A_l}{\comm{A_k}{A_{j}}} \right) + \frac{1}{4} \sum_{j>k}^{I-1} \comm{A_I}{\comm{A_j}{A_k}} \nonumber \\
			&+ \frac{1}{12} \sum_j^{I-1} \comm{A_I}{\comm{A_I}{A_j}} + \frac{1}{12} \sum_{j,k}^{I-1} \comm{A_j}{\comm{A_k}{A_I}} \nonumber \\
			&= \frac{1}{12} \sum_{j \neq k}^{I-1} \comm{A_j}{\comm{A_j}{A_k}} + \frac{1}{6} \sum_{j>k>l}^{I-1} \left( \comm{A_j}{\comm{A_k}{A_{l}}} + \comm{A_l}{\comm{A_k}{A_{j}}} \right) + \frac{1}{4} \sum_{j>k}^{I-1} \comm{A_I}{\comm{A_j}{A_k}} \nonumber \\
			&+ \frac{1}{12} \sum_j^{I-1} \left( \comm{A_I}{\comm{A_I}{A_j}} + \comm{A_j}{\comm{A_j}{A_I}} \right) + \frac{1}{12} \sum_{j>k}^{I-1} \left( - \comm{A_I}{\comm{A_j}{A_k}} + \comm{A_k}{\comm{A_j}{A_I}} \right) + \frac{1}{12} \sum_{j<k}^{I-1} \comm{A_j}{\comm{A_k}{A_I}} \nonumber \\
			&= \frac{1}{12} \sum_{j \neq k}^I \comm{A_j}{\comm{A_j}{A_{k}}} + \frac{1}{6} \sum_{j>k>l} \left( \comm{A_j}{\comm{A_k}{A_{l}}} + \comm{A_l}{\comm{A_k}{A_{j}}} \right),
		\end{align}
		where we first split the sum $\sum_{j,k} = \sum_{j=k} + \sum_{j>k} + \sum_{j<k}$ and used the Jacobi-identity of the commutator in the first step. In the second, we renamed indices to absorb the fourth sum into the second terms and the rest into the first term accordingly. Taking all three orders together, we conclude with the statement. 
	\end{proof}
\end{lemma}
We can now apply Lemma \ref{lem:BCH_trunc} to the BCH Hamiltonian $Z$ of a variational decomposition of $H$. While we will need the third order term of Eq.~\eqref{eq:multi_BCH} when discussing error scalings, for now we only make use of the first and second order term
\begin{align}
	-itH - Z &= -it \sum_{j=1}^M c_j H_j + i \sum_{r, j} \theta_{r,j} H_j - \frac{1}{2} \sum_{(r,j)>(r',j')} \theta_{r,j} \theta_{r',j'} \comm{H_j}{H_{j'}} + \mathcal{O}(\theta^3) \nonumber \\
	&= i \sum_{j} \left( \sum_r \theta_{r,j} - t c_j \right) H_j - \frac{1}{2} \sum_{(r,j)>(r',j')} \theta_{r,j} \theta_{r',j'} \comm{H_j}{H_{j'}} + \mathcal{O}(\theta^3), \label{eq:error_pre}
\end{align}
where we implied an order on the index set $\mathcal{J} = \{(r,j)\}_{r,j}$. The specific ordering is not important for our discussion. Without loss of generality, we will choose the order $(r,j)<(r',j')$, if and only if $r<r'$ or $r=r'$ and $j<j'$. If one fixes the last layer of parameters as in Eq.~\eqref{eq:cond}, the linear contribution in Eq.~\eqref{eq:error_pre} will vanish. For the sake of clarity, we will still write $\theta_{R,j}$ in the following even though it is not a free parameter. Finally, we can write the second order term of the upper bound 
\begin{align}
	\norm{itH + Z} &= \norm{\frac{1}{2} \sum_{(r,j)>(r',j')} \theta_{r,j} \theta_{r',j'} \comm{H_j}{H_{j'}}} \nonumber \\
	&= \norm{\frac{1}{2} \sum_{j>j'} \left( \sum_r\theta_{r,j} \theta_{r,j'} \right) \comm{H_j}{H_{j'}} + \frac{1}{2} \sum_{j,j'} \left(\sum_{r>r'} \theta_{r,j} \theta_{r',j'} \right) \comm{H_j}{H_{j'}} } \nonumber \\
	&= \norm{ \frac{1}{2} \sum_{j>j'} \left[ \sum_r \left( \theta_{r,j} \theta_{r,j'} \right) + \sum_{r>r'} \left( \theta_{r,j} \theta_{r',j'} - \theta_{r,j'} \theta_{r',j} \right) \right] \comm{H_j}{H_{j'}} }.
	\label{eq:error_gen}
\end{align}
We dropped vanishing commutators and regathered similar terms in the last step of Eq.~\eqref{eq:error_gen}. For the sake of clarity, we will define an abbreviation for the coefficient of the commutator terms
\begin{align}
	\chi_{j,j'} := \frac{1}{2} \sum_r \theta_{r,j} \theta_{r,j'} + \frac{1}{2} \sum_{r>r'} \left( \theta_{r,j} \theta_{r',j'} - \theta_{r,j'} \theta_{r',j} \right).
\end{align}
After plugging in Eq.~\eqref{eq:cond}, that is necessary to make the linear contribution in Eq.~\eqref{eq:error_pre} vanish, we get
\begin{align}
	\chi_{j,j'} = \frac{1}{2} t^2 c_j c_{j'} - t c_{j'} \sum_r^{R-1} \theta_{rj} + \sum_r^{R-1} \theta_{r,j} \theta_{r,j'} + \sum_{r>r'}^{R-1} \theta_{r,j} \theta_{r',j'}.
\end{align}

\subsection{Optimization potential (Proof of Proposition \ref{theo:opt_pot})}
\label{app:proof_pot}
\begin{proof}
	The backwards direction is trivial as $C(\theta)\geq 0$ per definition. As for the forward direction, assume that $\theta^T$ is a global minimum of $C$ and hence also of $C^2$. The cost value at the Trotter parameters $\theta^T_{r,j} = \frac{t c_j}{R}$ reads
	\begin{align}
		C\left(\theta^T\right)^2 = - 2^{-n} \sum_{j>j'} \sum_{k>k'} \frac{t^4 c_j c_{j'} c_k c_{k'}}{4R^2} \Tr\left( \comm{H_j}{H_{j'}} \comm{H_k}{H_{k'}} \right),
	\end{align}
	which is non-zero in general. For a global minimum, the gradient of $C$ at the point $\theta^T$ must vanish. The gradient of $C$ at the Trotter solution reads
	\begin{align}
		\partial_{q,l} C(\theta)^2 &= - 2^{-n+1} \sum_{j>j'} \sum_{k>k'} \left(\partial_{q,l} \chi_{j,j'}\right) \chi_{k,k'} \Tr\left( \comm{H_j}{H_{j'}} \comm{H_k}{H_{k'}} \right).
		\label{eq:C_grad_app}
	\end{align}
	Here, we introduces the shorthand $\partial_{r,j} = \frac{\partial}{\partial \theta_{r,j}}$. The factor 2 in Eq.~\eqref{eq:C_grad_app} comes from the fact that the cost function is symmetric under the exchange of the factors $\chi_{j,j'}$ and $\chi_{k,k'}$. The derivative of $\chi$ reads
	\begin{align}
		2\partial_{q,l} \chi_{j,j'} &= \delta_{l,j} \theta_{q,j'} + \delta_{l,j'} \theta_{q,j} + \sum_{q>r'} \left( \delta_{l,j} \theta_{r',j'} - \delta_{l,j'} \theta_{r',j} \right) + \sum_{r>q} \left( \theta_{r,j} \delta_{l,j'} - \theta_{r,j'} \delta_{l,j} \right) \nonumber \\
		&= \delta_{l,j} \left( \sum_{r \leq q} \theta_{q,j'} - \sum_{r>q} \theta_{q,j'} \right) + \delta_{l,j'} \left( \sum_{r<q} \theta_{q,j} - \sum_{r\geq q} \theta_{q,j} \right)
	\end{align}
	and at the Trotter solution
	\begin{align}
		2\partial_{q,l} \chi_{j,j'} &= \delta_{l,j} \frac{t c_{j'}}{R} (2q - R) + \delta_{l,j'} \frac{t c_j}{R} (2(q-1) - R).
	\end{align}
	Finally, the gradient of the cost function reads
	\begin{align}
		\left( \partial_{q,l} C^2 \right) \left( \left\{\frac{t c_j}{R}\right\}_j\right) = &-2^{-n} \sum_{j<l} \sum_{k>k'} \frac{t c_j}{R} (2q - R) \frac{t^2 c_k c_{k'}}{2R} \Tr\left( \comm{H_l}{H_{j}} \comm{H_k}{H_{k'}} \right) \nonumber \\
		&+2^{-n} \sum_{j>l} \sum_{k>k'} \frac{t c_{j}}{R} (2(q-1) - R) \frac{t^2 c_k c_{k'}}{2R} \Tr\left( \comm{H_l}{H_j} \comm{H_k}{H_{k'}} \right) \nonumber \\
		=& 2^{-n} \left((2(q-1) - R) \sum_{j>l} - (2q - R) \sum_{j<l} \right) \frac{t c_{j}}{R} \sum_{k>k'} \frac{t^2 c_k c_{k'}}{2R} \Tr\left( \comm{H_l}{H_j} \comm{H_k}{H_{k'}} \right).
	\end{align}
	We have $\nabla C\left(\theta^T\right) = 0$ if and only if
	\begin{align}
		\sum_{j>l} \sum_{k>k'} c_j c_k c_{k'} \Tr\left( \comm{H_l}{H_j} \comm{H_k}{H_{k'}} \right) = \frac{2q - R}{2(q-1) - R} \sum_{j<l} \sum_{k>k'} c_j c_k c_{k'} \Tr\left( \comm{H_l}{H_j} \comm{H_k}{H_{k'}} \right) \quad \forall q,l
		\label{eq:C_grad0}
	\end{align}
	Since $R\geq 3$, the right hand side of Eq.~\eqref{eq:C_grad0} will have a non-trivial dependence on $q$, but the left hand side will not. For instance, plug in $q=1$ and $q=R$ with arbitrary but fixed $l$, then the prefactor on the right hand side becomes $\frac{R-2}{R}$ and $\frac{R}{R-2}$. As for $R\geq3$, these are unequal, so the left hand side of Eq.~\eqref{eq:C_grad0} must vanish. Since $l$ is arbitrary, also $C(\theta)=0 \, \forall t$ which means Trotterization is exact.
\end{proof}

\subsection{Translational Invariant Models (Proof of Proposition \ref{theo:trans_inv})}
\label{app:trans_inv}
For the following proofs, we are interested in the leading order term $E(\theta)$ in the error from the truncation to the perturbative regime. This term will become dominant by the time, we have minimized $C(\theta)$. As before in appendix \ref{app:BCH_trunc}, we expand $\norm{itH + Z}$ now up to third order and separate $C$ from the third order error term $E$ using the triangle inequality
\begin{align}
	\norm{itH + Z} \leq C(\theta) + E(\theta) + \mathcal{O}\left( \theta ^4 \right).
\end{align}
Using Lemma \ref{lem:BCH_trunc}, we can read off
\begin{align}
	E(\theta) &= \norm{ \frac{1}{12} \sum_{(r,j) \neq (r',j')} \theta_{r,j} \theta_{r,j} \theta_{r',j'} \comm{H_j}{\comm{H_j}{H_{j'}}} + \frac{1}{6} \sum_{(r,j)>(u,k)>(v,l)} \theta_{r,j} \theta_{u,k} \theta_{v,l} \left( \comm{H_j}{\comm{H_k}{H_{l}}} + \comm{H_l}{\comm{H_k}{H_{j}}} \right) } \nonumber \\
	&= \Bigg\Vert \frac{1}{12} \sum_{j \neq j'} \left( \sum_r \theta_{r,j} \theta_{r,j} \theta_{r,j'} + \sum_{r\neq r'} \theta_{r,j} \theta_{r,j} \theta_{r',j'} \right) \comm{H_j}{\comm{H_j}{H_{j'}}} + \frac{1}{6} \Bigg( \sum_{j>k>l} \sum_r \theta_{r,j} \theta_{r,k} \theta_{r,l} \nonumber \\
	&+ \sum_{\substack{j>k\\l}} \sum_{r>v} \theta_{r,j} \theta_{r,k} \theta_{v,l} + \sum_{\substack{k>l\\j}} \sum_{r>u} \theta_{r,j} \theta_{u,k} \theta_{u,l} + \sum_{j,k,l} \sum_{r>u>v} \theta_{r,j} \theta_{u,k} \theta_{v,l} \Bigg) \left( \comm{H_j}{\comm{H_k}{H_{l}}} + \comm{H_l}{\comm{H_k}{H_{j}}} \right) \Bigg\Vert
	\label{eq:E}
\end{align}
where we split the sums as before using an implicit row-major ordering on the multiindex $(r,j)$. For the sake of simplicity, we gather all coefficients into a tensor $\Phi_{j,k,l}$ of the two-fold nested commutators and define $E(\theta) = \norm{\sum_{j,k,l} \Phi_{j,k,l} \comm{H_j}{\comm{H_k}{H_{l}}}}$.

\begin{proof}[Proof of Proposition \ref{theo:trans_inv}]
	If we impose translational symmetry on the system of interest, we can split such a symmetric, $p$-local ($p$-nearest neighbor) Hamiltonian into $A$ different interaction types that act the same way on every qubit, i.e.
	\begin{align}
		H = \sum_{a=1}^A \sum_{j=1}^n c_a H^{(j)}_{a}.
	\end{align}
	The index $(j)$ is labeling the qubit number and $a$ the interaction type. The Hamiltonians $H_a^{(j)}$ have non-trivial support on the $p$ qubits starting with $j$. In one dimension, these are $j, j+1,...,j+p-1$. For higher dimensions, different interaction types may admit support on clusters of at most $p$ neighboring qubits. We do not make any assumptions on lattice geometry or dimension in the following and formally denote by $\mathcal{N}_p(j)$ the set of Hamiltonian terms that correspond to a $p$-nearest neighbor of the $j^\text{th}$ term. The parameters are chosen to be translational invariant by just dropping the dependence on the qubit number, i.e. $\theta_{r,j} = \theta_{r,a}$. As a consequence, also $\chi_{j,j'} = \chi_{a,a'}$ and $\Phi_{j,k,l} = \Phi_{a,b,c}$.
	
	Since we introduced another index $j \to (a,j)$, we need to refine the chosen ordering. Without loss of generality, we choose
	\begin{align}
		(r,a,j): \quad (1,1,1) \leq (1,1,j) \leq (1,a,j) \leq (r,a,j)
	\end{align}
	To simplify the cost function defined in Eq.~\eqref{eq:cost}, we decompose the sum 
	\begin{align}
		\sum_{(a,j)>(a',j')} = \sum_{\substack{j>j' \\ a=a'}} + \sum_{a>a'} \sum_{j,j'}.
	\end{align}
	We plug in the above structure into the definition of the cost function,  and use the locality of the Hamiltonian to get
	\begin{align}
		C(\theta) &= \norm{ \sum_{j>j'} \chi_{a,a} \comm{H_a^{(j)}}{H_a^{(j')}} + \sum_{a>a'} \sum_{j,j'} \chi_{a,a'} \comm{H_a^{(j)}}{H_{a'}^{(j')}} } \nonumber \\
		&= \norm{ \sum_j \left[ \sum_{ \substack{j' \in \mathcal{N}_p(j) \\ j'>j  }} \chi_{a,a} \comm{H_a^{(j')}}{H_a^{(j)}} + \sum_{a>a'} \sum_{j' \in \mathcal{N}_p(j)} \chi_{a,a'} \comm{H_a^{(j)}}{H_{a'}^{(j')}} \right] } \nonumber \\
		&\leq n \norm{ \sum_{ \substack{j' \in \mathcal{N}_p(j) \\ j'>j  }} \chi_{a,a} \comm{H_a^{(j')}}{H_a^{(j)}} + \sum_{a>a'} \sum_{j' \in \mathcal{N}_p(j)} \chi_{a,a'} \comm{H_a^{(j)}}{H_{a'}^{(j')}} } \quad \text{for one fixed } j \nonumber \\
		&=: n C_{\rm unit}(\theta),
	\end{align}
	where we used the triangle inequality in the last step. From translational invariance we know that the second last term yields the same value for every $j$ yielding a factor $n$. If the commutators $\comm{H_j}{H_{j'}}$ are orthogonal with respect to the Hilbert-Schmidt product, the sum $\sum_j$ can be pulled out of the squared norm yielding $n$ equal addends and hence $C(\theta) = \sqrt{n} C_{\rm unit}(\theta)$.
	
	The same way, we can reduce the computation of $E(\theta)$ to a unit cell. However, as we consider two-fold nested commutators here, the unit cell is of size $(3p-2)^D$
	\begin{align}
		E(\theta) &= \norm{ \sum_j \left[ \sum_{a,b,c} \sum_{k \in \mathcal{N}_p(j)} \sum_{l \in \mathcal{N}_p(k)} \Phi_{a,b,c} \comm{H^{(j)}_a}{\comm{H^{(k)}_b}{H^{(l)}_c}} \right] } \nonumber \\
		&\leq n E_{\rm unit}(\theta)
	\end{align}
	Again, if the commutators $\comm{H^{(j)}_a}{\comm{H^{(k)}_b}{H^{(l)}_c}}$ are orthogonal with respect to the Hilbert-Schmidt inner product, we get $E(\theta) = \sqrt{n} E_{\rm unit}(\theta)$.
\end{proof}

\subsection{Long simulation times (Proof of Proposition \ref{theo:long_time})}
\label{app:long_times}
\begin{proof}
	As Eq.~\eqref{eq:Uvar_long} is equivalent to the ansatz before by just mapping $R \to KR$, $t \to Kt$ and identifying $\theta_{r,j} = \theta_{r \% R,j}$, the cost function will have the same form as before with modified $\chi$ 
	\begin{align}
		\chi_{j,j'}^{(KR)} &= \frac{1}{2} \sum_{r=1}^{KR} \left( \theta_{r,j} \theta_{r,j'} \right) + \frac{1}{2} \sum_{r>r'}^{KR} \left( \theta_{r,j} \theta_{r',j'} - \theta_{r,j'} \theta_{r',j} \right) \nonumber \\
		&= \frac{K}{2} \sum_{r=1}^{R} \left( \theta_{r,j} \theta_{r,j'} \right) + \frac{K}{2} \sum_{r>r'}^{R} \left( \theta_{r,j} \theta_{r',j'} - \theta_{r,j'} \theta_{r',j} \right) = K \chi_{j,j'}^{(R)},
		\label{eq:chi_long}
	\end{align}
	where we denoted the number of layers as a superscript $\chi^{(R)}$ to better differentiate the two coefficients. The factor $K$ comes from the identification of $\theta_{r,j} = \theta_{r \% R,j}$. While the first factor is trivial, the second comes from splitting the antisymmetric sum into 
	\begin{align}
		\sum_{r>r'}^{KR} = \frac{K(K+1)}{2} \sum_{r>r'}^{R} + \frac{(K-1)K}{2}\sum_{r<r'}^{R} = \frac{K(K+1)}{2} \sum_{r>r'}^{R} - \frac{(K-1)K}{2}\sum_{r>r'}^{R} = K \sum_{r>r'}^{R}
	\end{align}
\end{proof}
In the same way, we can plug in the duplicated parameters $\theta^K$ into $E(\theta)$. We get the following identities
\begin{align}
	\sum_{r}^{KR} = K \sum_{r}^R \qquad \text{and} \qquad \sum_{r \neq r'}^{KR} = K^2 \sum_{r\neq r'}^R + K(K-1) \sum_{r=r'}^R,
\end{align}
which yields a global factor $K^2$ for the first sum $\sum_{j\neq j'}$ in Eq.~\eqref{eq:E}. For the second sum, unfortunately we do not have a common factor to pull out. We have
\begin{align}
	\sum_{r>u}^{KR} &= \frac{K(K+1)}{2} \sum_{r>u}^R + \frac{K(K-1)}{2} \sum_{r\leq u}^R = \sum_{r>u}^R + \frac{K(K-1)}{2} \sum_{r, u}^R \nonumber \\
	\sum_{r>u>v}^{KR} &= K \sum_{r>u>v}^R + \frac{K(K-1)(K+1)}{6} \sum_r^R \sum_{u>v}^R + \frac{K(K-1)(K-2)}{6} \sum_r^R \sum_{u \leq v}^R \nonumber \\
	&= K \sum_{r>u>v}^R + \frac{K(K-1)}{2} \sum_r^R \sum_{u>v}^R.
\end{align}
In the last step we used that the sum over $\sum_{r,u,v}$ vanishes because of antisymmetry of the commutator. While some of the above contributions multiply the error before duplication by a factor, others are non-zero extra terms. We thus get an upper bound
\begin{align}
	E(\theta^K) \leq K E(\theta) + \frac{K(K-1)}{2} \Lambda
\end{align}
involving the rest term
\begin{align}
	\Lambda = \Bigg\Vert &\left( \sum_{\substack{j>k \\l}} \sum_{r>u} - \sum_{\substack{j>k \\l}} \sum_{r\leq u} -2 \sum_{\substack{j\leq k \\l}} \sum_{r\leq u} \right) \theta_{rj} \theta_{uk} \theta_{vl} \left( \comm{H_j}{\comm{H_k}{H_{l}}} + \comm{H_k}{\comm{H_j}{H_{j}}} \right) \nonumber \\
	&+ 3 \sum_{j,k,l} \sum_{r>u>v} \theta_{rj} \theta_{uk} \theta_{vl} \left( \comm{H_k}{\comm{H_j}{H_{j}}} \right) \Bigg\Vert.
\end{align}

\section{Perturbative Distance of a Transverse field Ising model}
\begin{figure}[t]
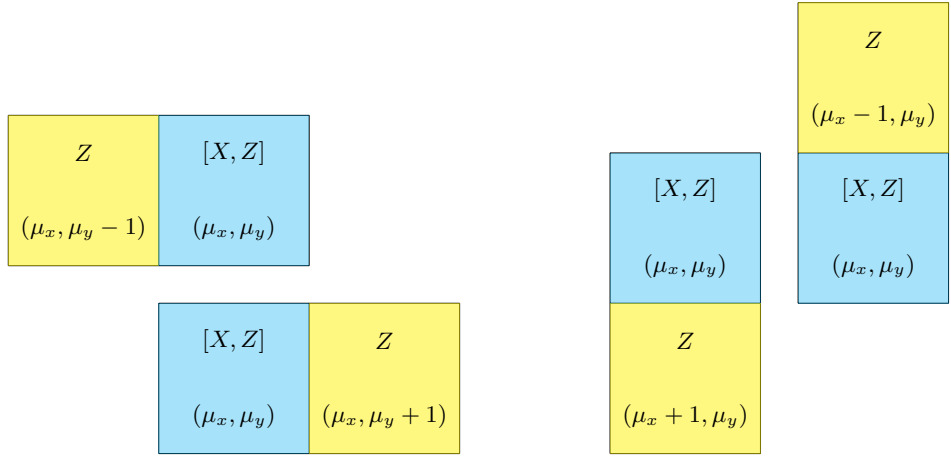

	\centering
	\horiz
	\caption{\textbf{Visualization of the perturbative distance} for a two-dimensional nearest neighbor Ising model. Instead of defining cumbersome notation as in Eq.~\eqref{eq:cost_TFIM2D}, it is often simpler to think about the commutator structure in this pictorial manner.}
	\label{fig:TFIM_terms}
\end{figure}

\label{app:Ising}
We apply the presented method to a general Transverse field Ising model (TFIM) with full connectivity, which is described by the Hamiltonian
\begin{align}
	H_{TFIM} = - \sum_{\mu,\nu} J_{\mu\nu} Z_\mu Z_\nu + \sum_\mu h_\mu X_\mu.
	\label{eq:Ham_TFIM}
\end{align}
The sum over $\mu,\nu$ counts every pair of qubits and the symmetric matrix $J_{\mu\nu}$ encodes the interaction strengths between qubits $\mu$ and $\nu$. To fix a time scale, we normalize the Hamiltonian to $\norm{H} \approx \sqrt{n}$ In total the Hamiltonian admits $\frac{n(n-1)}{2}$ ZZ terms and $n$ X terms. To use the notation introduced before, one needs to map the interactions between two qubits to a single index $j$
\begin{align}
	H_j = \begin{cases} 
		Z_{\mu(j)} Z_{\nu(j)} & \text{for } j \leq \frac{n(n-1)}{2} \\
		X_{\mu(j)} & \text{for } \frac{n(n-1)}{2} < j \leq \frac{n(n+1)}{2}
	\end{cases}
	\qquad 
	c_j = \begin{cases} 
		J_{\mu(j)\nu(j)} & \text{for } j \leq \frac{n(n-1)}{2} \\
		h_{\mu(j)} & \text{for } \frac{n(n-1)}{2} < j \leq \frac{n(n+1)}{2}
	\end{cases}
\end{align}
The exact ordering of the indices is not important, in general. To give an explicit example, consider row major ordering $j=(\mu-1)\left( n-\frac{\mu}{2} \right) + \nu$, for $\mu \leq \nu$. We plug in the $H_j$ into Eq.~\eqref{eq:cost} and observe that the only non-zero commutator terms are of the form $Z_\mu Y_\nu$. To also have a non-vanishing trace (Pauli-string squares to $\mathds{1}$), both commutators need to match, i.e.
\begin{align}
	\Tr\left( \comm{H_j}{H_{j'}} \comm{H_k}{H_{k'}} \right) = - 2^{n+2} \delta_{j,k} \delta_{j',k'} (\delta_{\mu(j), \mu(j')} + \delta_{\mu(j), \nu(j')}) \Theta_{j-\frac{n(n-1)}{2}} \Theta_{\frac{n(n-1)}{2}+1-j'},
\end{align}
where we used the implicit one-to-one mapping between indices $j \in \{1, ..., \frac{n(n+1)}{2}\}$ and $\mu, \nu \in \{1, ..., n\}$, as well as the Kronecker delta and the discretized Heavyside step function
\begin{align}
	\delta_{j,k} = \begin{cases} 1 & \text{if } j=k \\ 0 & \text{else} \end{cases} \qquad \Theta_{x} = \begin{cases} 1 & \text{if } x>0 \\ 0 & \text{else} \end{cases}.
\end{align}
The Heavyside factors ensure that only ZZ and X terms are composed in the commutator. The Kronecker deltas in the bracket account for the two cases where $X$ can hit the first or the second qubit of $ZZ$. The factors $\delta_{j,k} \delta_{j',k'}$ ensure that the Pauli-strings match yielding a non-vanishing trace. In total, we get
\begin{align}
	C(\theta)^2 &= 4 \sum_{j>j'} \sum_{k>k'} \chi_{j,j'} \chi_{k,k'} \delta_{j,k} \delta_{j',k'} (\delta_{\mu(j), \mu(j')} + \delta_{\mu(j), \nu(j')}) \Theta_{j-\frac{n(n-1)}{2}} \Theta_{\frac{n(n-1)}{2}+1-j'} \nonumber \\
	&= 4 \sum_{j=M_z+1}^{M_z+n} \sum_{j' = 1}^{M_z} \chi_{j,j'}^2 (\delta_{\mu(j), \mu(j')} + \delta_{\mu(j), \nu(j')}) \nonumber \\
	&= 4 \sum_{\mu < \nu}^{n} \left( \chi_{j(\mu),j(\mu, \nu)}^2 + \chi_{j(\nu),j(\mu, \nu)}^2  \right)
	\label{eq:cost_TFIM}
\end{align}
where we shorthanded the number of ZZ interactions as $M_z = \frac{n(n-1)}{2}$. Note that every addend in Eq.~\eqref{eq:cost_TFIM} is positive, so minimizing $C(\theta)$ reduces to the simultaneous minimization of each addend. However, the equations $\chi_{j,j'}^2 = 0$ do not decouple in general, so there is no exact formula for the roots of Eq.~\eqref{eq:cost_TFIM}.

\subsection{Two dimensional square lattice}
Physical models typically yield local interactions. Let us reduce the above discussion to a two-dimensional TFIM that only includes nearest neighbor interactions. This is straight-forwardly done by imposing
\begin{align}
	J_{\mu\nu} = 0 \qquad \text{and} \qquad \theta_{r,j(\mu, \nu)} = 0 \qquad \text{if } \mu \nsim \nu,
\end{align}
where $\mu \nsim \nu$ denotes that $\mu$ and $\nu$ are not neighbors. On a two-dimensional square lattice, it is handy to further refine the qubit indices into two $\mu=(\mu_x, \mu_y)$, where $\mu_x \in \{1, ..., n_x\}$ and $\mu_y \in \{1, ..., n_y\}$ count the lattice sites in each dimension, respectively. $\theta_{r,j} = 0$ also implies $\chi_{j,j'}=0 \, \forall j'$ which further reduces the number of terms in Eq.~\eqref{eq:cost_TFIM}.
\begin{align}
	C(\theta) &= 4 \sum_{\mu=(1,1)}^{(n_x,n_y)} \left( \chi_{j(\mu),j(\mu, (\mu_x, \mu_y+1))}^2 + \chi_{j((\mu_x, \mu_y+1)),j(\mu, (\mu_x, \mu_y+1))}^2 + \chi_{j(\mu),j(\mu, (\mu_x+1, \mu_y))}^2 + \chi_{j((\mu_x+1, \mu_y)),j(\mu, (\mu_x+1, \mu_y))}^2 \right)
	\label{eq:cost_TFIM2D}
\end{align}
Although the cost function becomes tedious to write down, it only contains a linear number of terms and is thus efficiently computable. The four contributions in Eq.~\eqref{eq:cost_TFIM} account for the different combinations of the appearing $ZY$ string and are depicted in Fig.~\ref{fig:TFIM_terms}.

\subsection{Numerical restults}
In Fig. \ref{fig:Ising3x3}, we study the performance of the derived cost function by optimizing a variational sequence with the second order cost function derived in Eq.~\eqref{eq:cost_TFIM2D} and comparing it to the exact value of the difference of time evolutions $\epsilon$ for different times $t$. The optimizer manages to find small values of $C(\theta)$ for all times. For small enough times $t < \frac{R}{\norm{H}} \approx 1$, the perturbative cost function also incorporates a faithful indicator for the exact norm in a sense that the optimal parameters also yield a significant decrease in the exact error of the variational sequence.

\begin{figure}[t]
	\centering
	\includegraphics[scale=0.7]{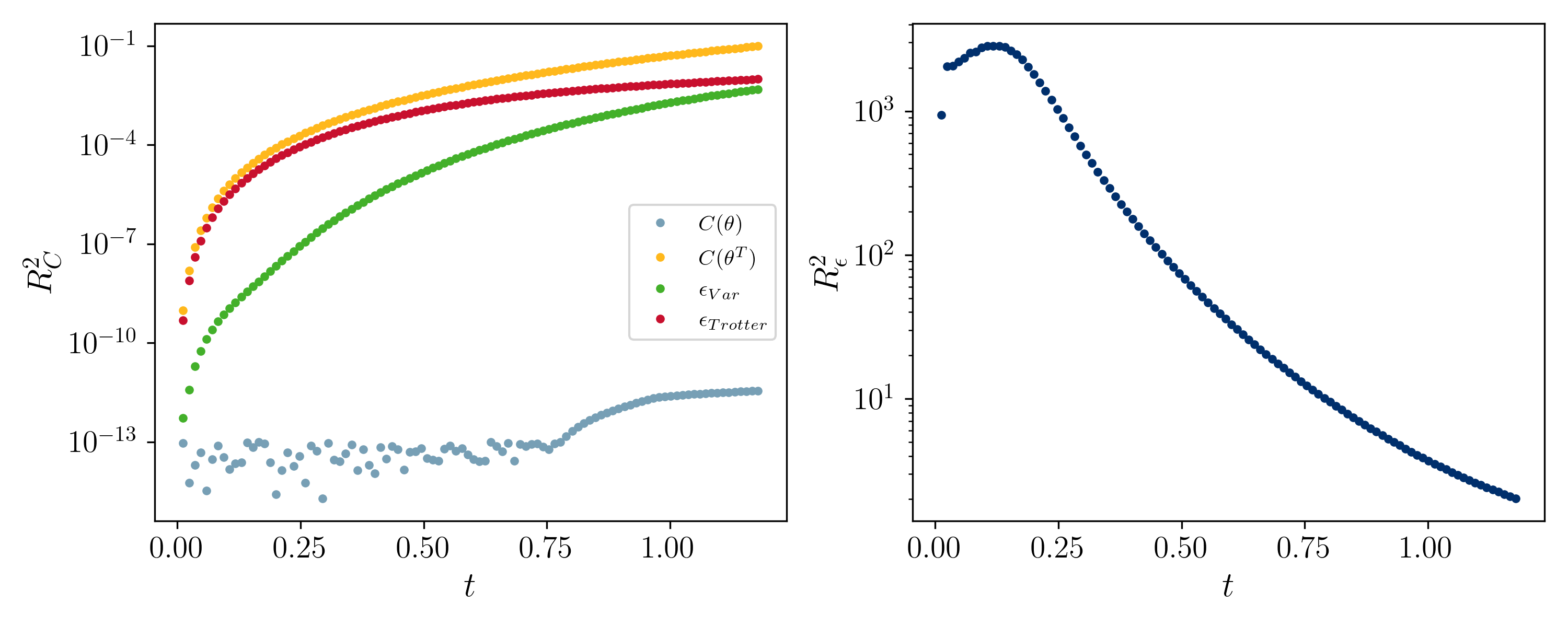}
	\caption{\textbf{Perturbative ($C$) and exact ($\epsilon$) error measures (left) and exact error ratio $R_\epsilon$ and the ratio of error estimates $R_E$ (right)} as defined in Eq.~\eqref{eq:ratios} evaluated on a nearest neighbor Ising model on a $3\times3$ quadratic lattice and random interaction strength $J_{\mu,\nu}$ centered around 1 and $h_\mu=0.25$. }
	\label{fig:Ising3x3}
\end{figure}

\section{Perturbative Distance of an XY model in a Transverse Field}
\label{app:XY}
In total, the XY Hamiltonian from Eq.~\eqref{eq:Ham_XY} admits $\frac{n(n-1)}{2}$ YY terms, $\frac{n(n-1)}{2}$ ZZ terms and $n$ X terms. To use the notation introduced before, one needs to map the interactions between two qubits to a single index $j$
\begin{align}
	H_j = \begin{cases} 
		X_{\mu(j)} & \text{for } j \leq n \\
		Y_{\mu(j)} Y_{\nu(j)} & \text{for } n < j \leq \frac{n(n+1)}{2} \\
		Z_{\mu(j)} Z_{\nu(j)} & \text{for } \frac{n(n+1)}{2} < j \leq n^2
	\end{cases}
	\qquad 
	c_j = \begin{cases} 
		h_{\mu(j)} & \text{for } j \leq n \\
		J^{(y)}_{\mu(j)\nu(j)} & \text{for } n < j \leq \frac{n(n+1)}{2} \\
		J^{(z)}_{\mu(j)\nu(j)} & \text{for } \frac{n(n+1)}{2} < j \leq n^2
	\end{cases}
	\label{eq:IsingXY_cases}
\end{align}
The exact ordering of the indices is not important, in general. To give an explicit example, consider row major ordering $j=(\mu-1)\left( n-\frac{\mu}{2} \right) + \nu$, for $\mu \leq \nu$. We plug in the $H_j$ into Eq.~\eqref{eq:cost} and observe that the only non-zero commutator terms are of the form $Z_\mu Y_\nu$ and $Y_\mu X_\nu Z_\sigma$. To also have a non-vanishing trace (Pauli-string squares to $\mathds{1}$), both commutators always need to match. While this is always given if $j=k$ and $j'=k'$, now we can also have the case $\comm{Y_\mu Y_\nu}{X_\mu} = - \comm{Z_\mu Z_\nu}{X_\nu} = -2 i Z_\mu Y_\nu$
\begin{align}
	C(\theta)^2 &= 4 \sum_{\mu < \nu}^{n} \left( \chi_{j_y(\mu, \nu),j(\mu)}^2 + \chi_{j_y(\mu, \nu),j(\nu)}^2 + \chi_{j_z(\mu, \nu),j(\mu)}^2 + \chi_{j_z(\mu, \nu),j(\nu)}^2 - \chi_{j_z(\mu, \nu),j(\mu)} \chi_{j_y(\mu, \nu),j(\nu)} - \chi_{j_z(\mu, \nu),j(\nu)} \chi_{j_y(\mu, \nu),j(\mu)} \right) \nonumber \\
	&+ 4 \sum_{\mu<\nu<\sigma}^n \left( \chi_{j_z(\mu, \nu),j_y(\nu, \sigma)}^2 + \chi_{j_z(\mu, \sigma),j_y(\nu, \sigma)}^2 + \chi_{j_z(\nu, \sigma),j_y(\mu, \nu)}^2 + \chi_{j_z(\nu, \sigma),j_y(\mu, \sigma)}^2 + \chi_{j_z(\mu, \nu),j_y(\mu, \sigma)}^2 + \chi_{j_z(\mu, \sigma),j_y(\mu, \nu)}^2 \right)
	\label{eq:cost_XY}
\end{align}
where we mapped back the lattice indices $\mu, \nu$ to the single index $j$ by reverting the mapping for all cases in Eq.~\eqref{eq:IsingXY_cases}. The terms in the first line of Eq.~\eqref{eq:cost_XY} come from commutator terms of the form YZ and the terms in the second line from terms of the form XYZ. 

\subsection{Two dimensional square lattice}
If we only admit nearest neighbor interaction and set all other $J^{(y,z)}_{\mu \nu} = 0$, we get
\begin{align}
	C(\theta)^2 &= 4 \sum_{\mu=(1,1)}^{(n_x,n_y)} \Big( \chi_{j_y(\mu, (\mu_x, \mu_y+1)),j(\mu)}^2 + \chi_{j_y(\mu, (\mu_x, \mu_y+1)),j((\mu_x, \mu_y+1))}^2 + \chi_{j_y(\mu, (\mu_x+1, \mu_y)),j(\mu)}^2 + \chi_{j_y(\mu, (\mu_x+1, \mu_y)),j((\mu_x+1, \mu_y))}^2 \nonumber \\
	&+ \chi_{j_z(\mu, (\mu_x, \mu_y+1)),j(\mu)}^2 + \chi_{j_z(\mu, (\mu_x, \mu_y+1)),j((\mu_x, \mu_y+1))}^2 + \chi_{j_z(\mu, (\mu_x+1, \mu_y)),j(\mu)}^2 + \chi_{j_z(\mu, (\mu_x+1, \mu_y)),j((\mu_x+1, \mu_y))}^2 \nonumber \\
	&- \chi_{j_z(\mu, (\mu_x, \mu_y+1)),j(\mu)} \chi_{j_y(\mu, (\mu_x, \mu_y+1)),j((\mu_x, \mu_y+1))} - \chi_{j_z(\mu, (\mu_x+1, \mu_y)),j(\mu)} \chi_{j_y(\mu, (\mu_x+1, \mu_y)),j((\mu_x+1, \mu_y))} \nonumber \\
	&- \chi_{j_z(\mu, (\mu_x, \mu_y+1)),j((\mu_x, \mu_y+1))} \chi_{j_y(\mu, (\mu_x, \mu_y+1)),j(\mu)} - \chi_{j_z(\mu, (\mu_x+1, \mu_y)),j((\mu_x+1, \mu_y))} \chi_{j_y(\mu, (\mu_x+1, \mu_y)),j(\mu)} \nonumber \\
	&+ \chi_{j_z(\mu, (\mu_x, \mu_y+1)),j_y((\mu_x, \mu_y+1), (\mu_x, \mu_y+2))}^2 + \chi_{j_z(\mu, (\mu_x+1, \mu_y)),j_y((\mu_x+1, \mu_y), (\mu_x+2, \mu_y))}^2 \nonumber \\
	&+ \chi_{j_z(\mu, (\mu_x, \mu_y+1)),j_y(\mu, (\mu_x+1, \mu_y))}^2 + \chi_{j_z(\mu, (\mu_x, \mu_y+1)),j_y((\mu_x, \mu_y+1), (\mu_x+1, \mu_y+1))}^2 \nonumber \\
	&+ \chi_{j_z(\mu, (\mu_x, \mu_y+1)),j_y((\mu_x-1, \mu_y), \mu)}^2 + \chi_{j_z(\mu, (\mu_x, \mu_y+1)),j_y((\mu_x-1, \mu_y+1), (\mu_x, \mu_y+1))}^2 \nonumber \\
	&+ \chi_{j_z((\mu_x, \mu_y+1), (\mu_x, \mu_y+2)),j_y(\mu, (\mu_x, \mu_y+1))}^2 + \chi_{j_z((\mu_x+1, \mu_y), (\mu_x+2, \mu_y)),j_y(\mu, (\mu_x+1, \mu_y))}^2 \nonumber \\
	&+ \chi_{j_z(\mu, (\mu_x+1, \mu_y)),j_y(\mu, (\mu_x, \mu_y+1))}^2 + \chi_{j_z((\mu_x, \mu_y+1), (\mu_x+1, \mu_y+1)),j_y(\mu, (\mu_x, \mu_y+1))}^2 \nonumber \\
	&+ \chi_{j_z((\mu_x-1, \mu_y), \mu),j_y(\mu, (\mu_x, \mu_y+1))}^2 + \chi_{j_z((\mu_x-1, \mu_y+1), (\mu_x, \mu_y+1)),j_y(\mu, (\mu_x, \mu_y+1))}^2 \Big) \label{eq:cost_XY2D}
\end{align}
Although being tedious to write down, Eq.~\eqref{eq:cost_XY2D} admits a linear number of addends that follow simple systematics. The first two lines gather the 2D-analogues of terms 1 through 4 in Eq.~\eqref{eq:cost_XY} and lines 3 and 4 gather the former terms 5 and 6. The remaining terms are gathered in lines 5 to 10 and resemble every combination of non-trvial three-body commutators.

\begin{figure*}[t]
	\centering
	\includegraphics[scale=0.7]{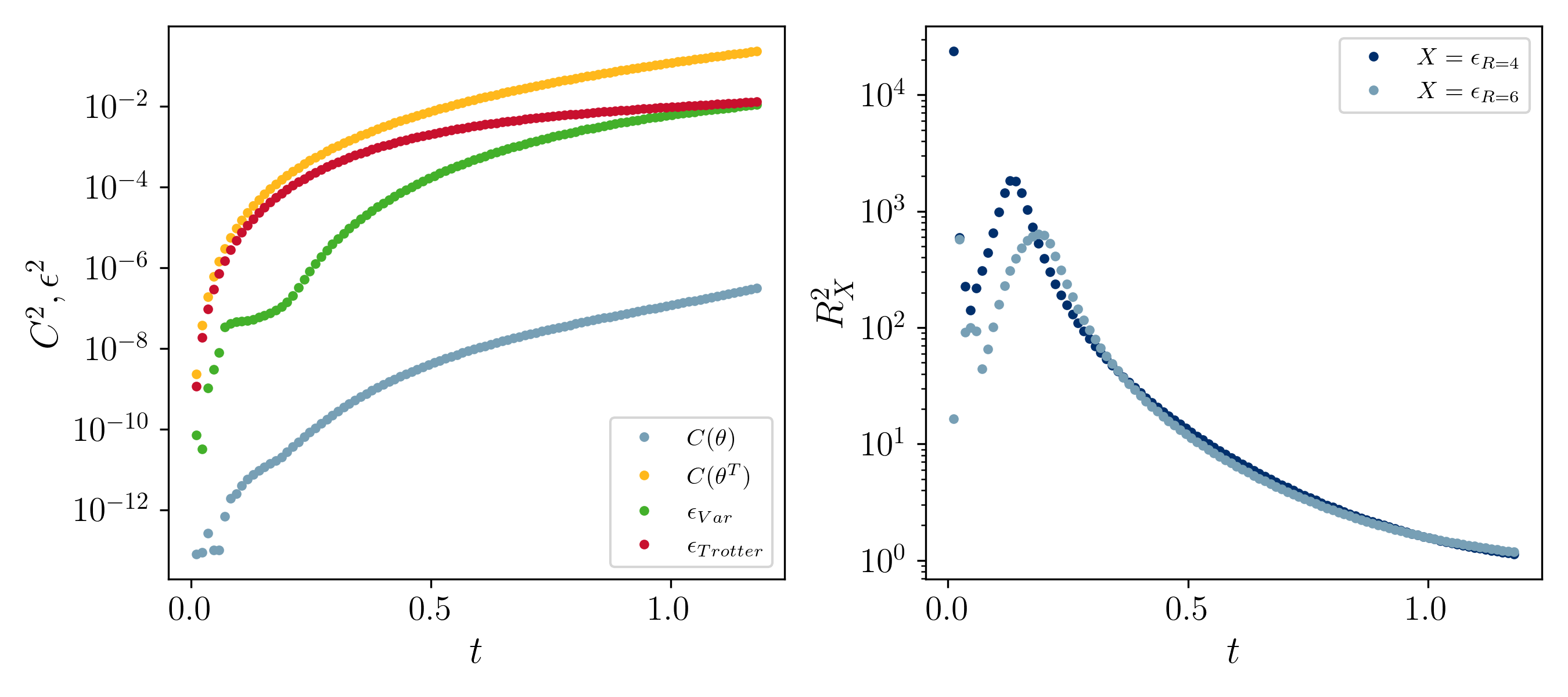}
	\caption{\textbf{Perturbative ($C$) and exact ($\epsilon$) error measures (left) and exact error ratio $R_\epsilon$ (right)} for $R=4$ and $R=6$ evaluated on a nearest neighbor XY model on a $3\times3$ quadratic lattice and random interaction strength $J^{(y)}_{\mu,\nu}$ centered around 0.5, $J^{(z)}_{\mu\nu}$ centered around 1 and $h_\mu=0.25$. On the left panel, results for $R=6$ are plotted.}
	\label{fig:IsingXY3x3_R46}
\end{figure*}
\subsection{More Parameters}
We explore the behaviour under changing the depth $R$ within a single time step, which also alters the number of parameters to be optimized. In Fig.~\ref{fig:IsingXY3x3_R46}, we show a numerical experiment analogue to Fig.~\ref{fig:IsingXY3x3}, but using $R=4$ and $R=6$. As with increasing $R$, $\frac{t}{R}$ decreases and the Trotter errors that we compare with get smaller. The optimizer still manages to find solutions that yield a comparable improvement ratio as seen in Fig.~\ref{fig:IsingXY3x3}. For very small times, the improvement ratios become much smaller, since the cost function $C$, as well as the error $\epsilon$, approach numerical accuracy. 

\end{document}